\documentclass[a4paper,12pt]{article}

\usepackage{a4wide}
\usepackage[usenames,dvipsnames]{color}
\usepackage{amssymb,amsmath,eepic,graphics,color,enumerate,eucal}
\usepackage{tikz}
\usetikzlibrary{calc}

\usepackage{epsfig}

\newcommand{\ignore}[1]{}%

\newcommand{\floor}[1]{\ensuremath{\lfloor{#1}\rfloor}}%

\newcommand{\threeEC}{$3$-edge-colorable}
\newcommand{\f}{\overline{\pi}}

\usepackage{amsmath,amsthm,amsfonts,amssymb,epsfig,color,xy}

\newtheorem{theorem}{Theorem}[]
\newtheorem{lemma}[theorem]{Lemma}
\newtheorem{corollary}[theorem]{Corollary}
\newtheorem{proposition}[theorem]{Proposition}
\newtheorem{observation}[theorem]{Observation}
\newtheorem{invariant}{Invariant}
\newtheorem{question}{Question}

\newcommand{\pred}{\textnormal{pred}}
\newcommand{\ch}{\textnormal{ch}}
\newcommand{\nfc}{\textnormal{nfc}}

\newcommand{\F}{\mathcal{F}}
\newcommand{\G}{\mathcal{G}}
\newcommand{\OPT}{{\rm OPT}}
\newcommand{\ovGamma}{\overline{\Gamma}}
\newcommand{\ovc}{\overline{c}}
\newcommand{\half}{\frac{1}2}
\newcommand{\mycase}[1]{\noindent{\bf CASE #1\ }}

\usepackage[T1]{fontenc}

\begin{document}
\title{Beyond the Vizing's bound for at most seven colors\thanks{A preliminary version of this work (with a proper subset of the results) was presented at 12th Scandinavian Symposium and Workshops on Algorithm Theory (SWAT 2010) and is published as an extended abstract~\cite{swat}.}}

\author{
  Marcin Kami\'{n}ski\footnote{D\'epartement d'Informatique, Universit\'e Libre de Bruxelles and Institute of Informatics, University of Warsaw. Email: {\tt mjk@mimuw.edu.pl} } \and \L{}ukasz Kowalik\footnote{Institute of Informatics, University of Warsaw. Email: {\tt kowalik@mimuw.edu.pl}. Supported by ERC StG project PAAl no.\ 259515.}
}

\date{}

\maketitle

{
\begin{abstract}
Let $G=(V,E)$ be a simple graph of maximum degree $\Delta$. The edges of $G$ can be colored with at most $\Delta +1$ colors by Vizing's theorem. We study lower bounds on the size of subgraphs of $G$ that can be colored with $\Delta$ colors. 

Vizing's Theorem gives a bound of $\frac{\Delta}{\Delta+1}|E|$. This is known to be tight for cliques $K_{\Delta+1}$ when $\Delta$ is even. However, for $\Delta=3$ it was improved to $\frac{26}{31}|E|$ by Albertson and Haas [\emph{Parsimonious edge colorings}, Disc. Math. 148, 1996] and later to $\frac{6}7|E|$ by Rizzi [\emph{Approximating the maximum 3-edge-colorable subgraph problem}, Disc. Math. 309, 2009]. It is tight for $B_3$, the graph isomorphic to a $K_4$ with one edge subdivided. 

We improve previously known bounds for $\Delta\in\{3,\ldots,7\}$, under the assumption that for $\Delta=3,4,6$ graph $G$ is not isomorphic to $B_3$, $K_5$ and $K_7$, respectively. For $\Delta \geq 4$ these are the first results which improve over the Vizing's bound.
We also show a new bound for subcubic multigraphs not isomorphic to $K_3$ with one edge doubled.

In the second part, we give approximation algorithms for the Maximum $k$-Edge-Colorable Subgraph problem, where given a graph $G$ (without any bound on its maximum degree or other restrictions) one has to find a $k$-edge-colorable subgraph with maximum number of edges. 
In particular, when $G$ is simple for $k=3,4,5,6,7$ we obtain approximation ratios of $\frac{13}{15},\frac{9}{11}$, $\frac{19}{22}$, $\frac{23}{27}$ and $\frac{22}{25}$, respectively. We also present a $\frac{7}{9}$-approximation for $k=3$ when $G$ is a multigraph. The approximation algorithms follow from a new general framework  that can be used for any value of~$k$.

\end{abstract}}


\newcommand{\heading}[1]{\medskip\noindent\textbf{#1}.\ } 
\newcommand{\headingnsp}[1]{\noindent\textbf{#1}.\ } 

\section{Introduction}

A graph is said to be $k$-edge-colorable if there exists an assignment of colors from the set $\{1,\ldots,k\}$ to the edges of the graph, such that every two incident edges receive different colors. 
For a graph $G$, let $\Delta(G)$ denote the maximum degree of $G$.
Clearly, we need at least $\Delta(G)$ colors to color all edges of graph $G$.
On the other hand, the celebrated Vizing's Theorem~\cite{vizing} states that for simple
graphs $\Delta+1$ colors always suffice.
However, if $k<\Delta+1$ it is an interesting question how many edges of $G$ can be colored in $k$ colors.
The maximum $k$-edge-colorable subgraph of $G$ (maximum $k$-ECS in short) is a $k$-edge-colorable subgraph $H$ of $G$ with maximum number of edges.
By $\gamma_k(G)$ we denote the ratio $|E(H)|/|E(G)|$; when $|E(G)|=0$ we define $\gamma_k(G)=1$.
The {\sc Maximum $k$-Edge-Colorable Subgraph} problem (aka Maximum Edge $k$-coloring~\cite{FOW02}) is to compute a maximum $k$-ECS of a given graph.
It is known to be APX-hard when $k \geq 2$ \cite{CP80, H81, LG83}. 

The research on approximation algorithms for max $k$-ECS problem was initiated by Feige, Ofek and Wieder~\cite{FOW02}. Among other results, they suggested the following simple strategy.
Begin with finding a maximum $k$-matching $F$ of the input graph, i.e.\ a subgraph of maximum degree $k$ which has maximum number of edges. This can be done in polynomial time (see e.g.~\cite{schrijver}). Since a $k$-ECS is a $k$-matching itself, $F$ has at least as many edges as the maximum $k$-ECS. Hence, if we color $\rho|E(F)|$ edges of $F$ we get a $\rho$-approximation. It follows that studying large $k$-edge-colorable subgraphs of graphs of maximum degree $k$ is particularly interesting. Let us conclude this paragraph by the following proposition.

\begin{proposition}[Feige, Ofek and Wieder~\cite{FOW02}]
 \label{prop:approx}
 If every graph $G=(V,E)$ of maximum degree $k$ has a $k$-edge-colorable subgraph with at least $\rho|E|$ edges,
 and such a subgraph can be found in polynomial-time, then there is a $\rho$-approximation algorithm for the maximum $k$-ECS problem.\qed
\end{proposition}

\subsection{Large $\Delta$-edge-colorable subgraphs of graphs of maximum degree $\Delta$}
\label{intro:combin}
As observed in~\cite{FOW02}, if we have a simple graph $G$ of maximum degree $\Delta(G)$, and we find its $(\Delta+1)$-edge-coloring by the algorithm which follows from the proof of Vizing's Theorem, we can simply choose the $\Delta$ largest color classes to $\Delta$-color at least $\frac{\Delta}{\Delta+1}$ edges of $G$.
Can we do better? In general we cannot, and the tight examples are the graphs $K_{\Delta+1}$, for even values of $\Delta$ (see Lemma~\ref{lem:even-clique}).
However, for odd values of $\Delta$ the best upper bound is $\frac{\Delta+1}{\Delta+2-1/\Delta}$ which is attained by graph $B_\Delta$ (see Lemma~\ref{lem:upper-odd}). This raises two natural questions.
 
\begin{question}
\label{q-odd}
When $\Delta$ is odd, can we obtain a better lower bound than $\frac{\Delta}{\Delta+1}$ for simple graphs?
\end{question}

\begin{question}
\label{q-even}
When $\Delta$ is even and $G\ne K_{\Delta+1}$, can we obtain a better lower bound than $\frac{\Delta}{\Delta+1}$  for simple graphs?
\end{question}

\headingnsp{Previous Work}
Question~\ref{q-odd} has been answered in affirmative for $\Delta=3$ by Albertson and Haas~\cite{AH96}, namely they showed that $\gamma_3(G)\ge\frac{26}{31}$ for simple graphs.
They also showed that $\gamma_3(G)\ge\frac{13}{15}$ when $G$ is cubic simple graph.
Recently, Rizzi~\cite{R09} showed that $\gamma_3(G)\ge\frac{6}{7}$ when $G$ is a simple subcubic graph. The bound is tight by a $K_4$ with an arbitrary edge subdivided (we denote it by $B_3$). Rizzi also showed that when $G$ is a multigraph with no cycles of length 3, then $\gamma_3(G)\ge\frac{13}{15}$, which is tight by the Petersen graph. We are not aware on any results for $\Delta$ bigger than $3$.

\heading{Our Contribution}
In the view of the result of Rizzi it is natural to ask whether $B_3$ is the only subcubic simple graph $G$ with $\gamma_3(G)=\frac{6}{7}$.
We answer this question in affirmative, namely we show that $\gamma_3(G)\ge\frac{13}{15}$ when $G$ is a simple subcubic graph different from $B_3$.
This generalizes both the bound of Rizzi for triangle-free graphs and the bound of Albertson and Haas~\cite{AH96} for cubic graphs, and is tight by the Petersen graph.
For a subcubic multigraph, the bound $\gamma_3(G)\ge\frac{3}{4}$ (Vizing's Theorem holds for subcubic multigraphs) is tight by the $K_3$ with an arbitrary edge doubled (we denote it by $G_3$).
Again, we show that $G_3$ is the only tight example: $\gamma_3(G)\ge\frac{7}{9}$ when $G$ is a subcubic multigraph different from $G_3$.

The two results mentioned above follow relatively fast from the work of Rizzi~\cite{R09}.
Our main technical contribution is the positive answer to Questions~\ref{q-odd} and~\ref{q-even} for $\Delta\in\{4,\ldots,7\}$.
Namely, we show that 
\begin{itemize}
 \item 
$\gamma_4(G)\ge\frac{5}{6}$ when $G$ is a simple graph of maximum degree 4 different from $K_5$,
 \item 
$\gamma_5(G)\ge\frac{23}{27}$ when $G$ is a simple graph of maximum degree 5,
 \item 
$\gamma_6(G)\ge\frac{19}{22}$ when $G$ is a simple graph of maximum degree 6 different from $K_7$,
 \item 
$\gamma_7(G)\ge\frac{22}{25}$ when $G$ is a simple graph of maximum degree 7.
\end{itemize}
In order to achieve the above bounds we develop a mini-theory describing the structure of maximum $\Delta$-edge-colorable subgraphs and their colorings, which may be useful for further research. 

Very recently Mkrtchyan and Steffen~\cite{steffen} showed that every simple graph $G$ has a maximum $\Delta(G)$-edge-colorable subgraph $H$ such that $E(G)\setminus E(H)$ is a matching. Hence, our bounds combined with this result can be seen as a strengthening of Vizing's theorem: e.g.\ we show that every graph of maximum degree 4 distinct from $K_5$ has a $5$-edge-coloring such that the $4$ largest color classes contain at least $\frac{5}{6}|E|$ edges.

\subsection{Approximation algorithms for the max $k$-ECS problem}

\noindent\textbf{Previous work}. 
As observed in~\cite{FOW02}, the $k$-matching technique mentioned in the beginning of this section together with the bound $\gamma_k(G)\ge\frac{k}{k+1}$ of Vizing's Theorem gives a $\frac{k}{k+1}$-approximation algorithm for simple graphs and every $k\ge 2$.
Note that the approximation ratio approaches 1 as $k$ approaches $\infty$.
For multigraphs, we get a $\frac{k}{k+\mu(G)}$-approximation by Vizing's Theorem and a 
$k/\lfloor\frac{3}{2}k\rfloor$-approximation by the Shannon's Theorem on edge-colorings~\cite{shannon}.

Feige et al.~\cite{FOW02} show a polynomial-time algorithm which, for a given multigraph and an integer $k$, finds a subgraph $H$ such that $|E(H)|\ge {\rm OPT}$, $\Delta(H)\le k+1$ and $\Gamma(H)\le k+\sqrt{k+1}+2$, where ${\rm OPT}$ is the number of edges in the maximum $k$-edge colorable sugraph of $G$, and $\Gamma(H)$ is the odd density of $H$, defined as $\Gamma(H)=\max_{S\subseteq V(H), |S|\ge 2}\frac{|E(S)|}{\lfloor |S|/2\rfloor}$. The subgraph $H$ can be edge-colored with at most $\max\{\Delta+\sqrt{\Delta/2},\lceil \Gamma(H)\rceil\}\le\lceil k + \sqrt{k+1} + 2 \rceil$ colors in $n^{O(\sqrt{k})}$-time by an algorithm of Chen, Yu and Zang~\cite{chen-jco}. By choosing the $k$ largest color classes as a solution this gives a $k/\lceil k + \sqrt{k+1} + 2 \rceil$-approximation. One can get a slightly worse $k/(k+(1+3/\sqrt{2})\sqrt{k} + o(\sqrt{k}))$-approximation by replacing the algorithm of Chen et al. by an algorithm of Sanders and Steurer~\cite{ss} which takes only $O(nk(n+k))$-time.
Note that in both cases the approximation ratio approaches 1 when $k$ approaches $\infty$, similarly as in the case of simple graphs.

The results above work for all values of $k$. However, for small values of $k$ tailor-made algorithms are known, with much better approximation ratios.
The most intensively studied case is $k=2$.
The research of this basic variant was initiated by Feige et al.~\cite{FOW02}, who proposed an algorithm for multigraphs based on an LP relaxation with an approximation ratio of $\frac{10}{13} \approx 0.7692$.
They also pointed out a simple $\frac{4}{5}$-approximation for simple graphs.
This was later improved several times~\cite{CTW08,CT09}. 
In 2009 Kosowski~\cite{K09} achieved a $\frac{5}{6}$-approximation by a very interesting extension of the $k$-matching technique (see Section~\ref{sec:OTW}). Finally, Chen, Konno and Matsushita~\cite{chen-0-842} got a $0.842$-approximation, essentially by a very careful analysis of the structure of the $k=2$ case. 
 
Kosowski~\cite{K09} studied also  the case of $k=3$ and obtained a $\frac{4}{5}$-approximation for simple graphs, which was later improved by a $\frac{6}{7}$-approximation resulting from the mentioned result of Rizzi~\cite{R09}.

Finally, there is a simple greedy algorithm by Feige et al.~\cite{FOW02} with approximation ratio $1-(1-\frac{1}k)^k$, which is still the best result for the case $k=4$ in multigraphs.

{
  
\linespread{1.1}
\begin{table}[t]
       \begin{center}
\begin{tabular}{|c||c|c||c|c|}
                       \hline
                       $k$ & simple graphs & reference & multigraphs & reference  \\
                       \hline                  \hline
                       2 & $0.842$ & \cite{chen-0-842} & $\frac{10}{13}$ & \cite{FOW02} \\
                       3 & $\frac{13}{15}$ & {\bf this work} & $\frac{7}{9}$ & {\bf this work} \\
                       4 & $\frac{9}{11}$ & {\bf this work} & $1-(\frac{3}{4})^4> 0.683$ & \cite{FOW02} \\
                       5 & $\frac{23}{27}$ & {\bf this work} & $\frac{5}{7}$ & \cite{shannon,FOW02} \\
                       6 & $\frac{19}{22}$ & {\bf this work} & $\max\{\frac{2}{3},\frac{6}{6+\mu}\}$ & \cite{shannon,vizing,FOW02} \\
                       7 & $\frac{22}{25}$ & {\bf this work} & $\max\{\frac{7}{10},\frac{7}{7+\mu}\}$ & \cite{shannon,vizing,FOW02} \\
                       $8, \ldots, 13$ & $\frac{k}{k+1}$ & \cite{vizing,FOW02} &  $\max\{\frac{k}{\lfloor 3k/2\rfloor},\frac{k}{k+\mu}\}$ & \cite{shannon,vizing,FOW02} \\
                       $\ge 14$ & $\frac{k}{k+1}$ & \cite{vizing,FOW02} &  $\max\{\frac{k}{\lceil k+\sqrt{k+1}+2\rceil},\frac{k}{k+\mu}\}$ & \cite{chen-jco,vizing,FOW02} \\
                       \hline
               \end{tabular}
      \end{center}
       \caption{\label{piekna-tabelka}Best approximation ratios for the Maximum $k$-Edge-Colorable Subgraph problem}
\end{table}
 
}

\heading{Our contribution}
We generalize the technique that Kosowski used in his algorithm for the max 2-ECS problem so that it may be applied for an arbitrary number of colors.
Roughly, we deal with the situation when for a graph $G$ of maximum degree $k$ one can find in polynomial time a $k$-edge colorable subgraph $H$ with at least $\alpha|E(G)|$ edges, unless $G$ belongs to a family $\F$ of ``exception graphs'', i.e. $\gamma(G)<\alpha$.
As we have seen in the case of $k=3,4,6$ the set of exception graphs is small and in the case of $k=2$ the exceptions form a very simple family of graphs (odd cycles). 
The exception graphs are the only obstacles which prevent us from obtaining an $\alpha$-approximation algorithm (for general graphs) by using the $k$-matching approach.
In such situation we provide a general framework, which allows to obtain approximation algorithms with approximation ratio better than $\min_{A\in\F}\gamma_k(A)$.
See Theorem~\ref{th:meta-algorithm} for the precise description of our general framework.

By combining the framework and our combinatorial results described in Section~\ref{intro:combin} we get the following new results (see Table~\ref{piekna-tabelka}):
a $\frac{7}{9}$-approximation of the max-$3$-ECS problem for multigraphs,
a $\frac{13}{15}$-approximation of the max-$3$-ECS problem for simple graphs, 
a $\frac{9}{11}$-approximation of the max-$4$-ECS problem for simple graphs,
a $\frac{23}{27}$-approximation of the max-$5$-ECS problem for simple graphs,
a $\frac{19}{22}$-approximation of the max-$6$-ECS problem for simple graphs, and
a $\frac{22}{25}$-approximation of the max-$7$-ECS problem for simple graphs.
Note that for 4 up to 7 colors our algorithms are the first which break the barrier of Vizing's Theorem.
Although we were able to get improved approximation ratios only for at most seven colors, note that these are the most important cases, since the approximation ratio of the algorithm based on Vizing's theorem is very close to 1 for large number of colors.

\subsection{Notation} 
We use standard terminology; for notions not defined here, we refer the reader to \cite{Diestel}. 
Let $G=(V,E)$ be an undirected graph.
For a vertex $x$ by $N_G(x)$ we denote the set of neighbors of $x$ and $N_G[x]=N_G(x)\cup\{x\}$. 
For a set of vertices $S$ we denote $N_G(S)=\bigcup_{x\in S}N_G(x)\setminus S$ and $N_G[S]=\bigcup_{x\in S}N_G[x]$.
Moreover, $d_G(S)=\{uv \in E\ :\ u\in S, v\not\in S\}$. 
For two sets $X,Y\subseteq V$ we define $E_G(X,Y)=\{xy \in E\ :\ x\in X\setminus Y, y\in Y\setminus X\}$.
In all of the above denotations we omit the subscripts when it is clear what graph we refer to.
A graph with maximum degree 3 is called \emph{subcubic}.
Following~\cite{AH96}, let $c_k(G)$ be the maximum number of edges of a $k$-edge-colorable subgraph of $G$. We also denote $\overline{c}_k(G)=|E(G)|-c_k(G)$, 
$c(G)=c_{\Delta(G)}(G)$ and $\overline{c}(G)=\overline{c}_{\Delta(G)}(G)$.

\section{Large 3-edge-colorable subgraphs of graphs maximum degree 3}

In this section we will work with multigraphs. We will also need the following result on triangle-free multigraphs from Rizzi \cite{R09}. 

\begin{lemma}[Rizzi \cite{R09}] \label{lemma-Rizzi}
Every subcubic, triangle-free multigraph $G$ has a 3-edge-colorable subgraph with at least $\frac{13}{15}|E(G)|$ edges. Moreover, this subgraph can be found in polynomial time.
\end{lemma}

We need one more definition. Let  $G^*_5$ be the multigraph on 5 vertices obtained from the four-vertex cycle by doubling one edge of the cycle and adding a vertex of degree two adjacent to the two vertices of the cycle not incident with the double edge.

\begin{theorem} \label{theorem-13-15}
Let $G$ be a biconnected subcubic multigraph different from $G_3$, $B_3$ and $G^*_5$. There exists a 3-edge-colorable subgraph of $G$ with at least $\frac{13}{15} |E(G)|$ edges. Moreover, this subgraph and its coloring can be found in polynomial time.
\end{theorem}

\begin{proof}
We will prove the theorem by induction on the number of vertices of the multigraph. We introduce the operation of \emph{triangle contraction} which is to contract the three edges of a triangle (order of contracting is inessential) keeping multiple edges that appear. 
Note that since $G$ is biconnected and $G\ne G_3$, no triangle in $G$ has
a double edge, so loops do not appear after the triangle contraction operation.
If a multigraph is subcubic, then it will be subcubic after a triangle contraction. Notice that if a multigraph has at least five vertices, the operation of triangle contraction in subcubic multigraphs preserves biconnectivity. It is easy to check that that all subcubic multigraphs on at most 4 vertices, different from $G_3$, have a 3-edge-colorable subgraph with least $\frac{13}{15} |E(G)|$ edges.

Let $G$ be a biconnected subcubic multigraph with at least 5 vertices and different from $B_3$. If $G$ is triangle-free, then the theorem follows from Lemma \ref{lemma-Rizzi}. Let us assume that $G$ has at least one triangle $T$ and let $G'$ be the multigraph obtained from $G$ by contracting $T$.

We can assume that $G'$ is subcubic and biconnected. First, let us assume that $G'$ is not isomorphic to $G_3$, $B_3$, or $G^*_5$. $G'$ has less vertices than $G$ so by the induction hypothesis it has a 3-edge-colorable subgraph with at least $\frac{13}{15} |E(G')|$. Notice that it can always be extended to contain all three edges of $T$. Hence, $G$ has a 3-edge-colorable subgraph with at least $\frac{13}{15} |E(G')| + 3 \geq \frac{13}{15} |E(G)|$ edges.

Now we consider the case when $G'$ is isomorphic to $G_3$, $B_3$ or $G^*_5$. In fact, $G'$ cannot be isomorphic to $G_3$, because then $G$ would be $B_3$ or $G^*_5$. There are only three multigraphs from which $B_3$ can be obtained after triangle contraction; they all have 10 edges and a \threeEC~subgraph with $9 > \frac{13}{15} \cdot 10$ edges. Similarly, there are only three multigraphs from which $G^*_5$ can be obtained after triangle contraction; they all have 10 edges and a \threeEC~subgraph with $9 > \frac{13}{15} \cdot 10$ edges. 
\end{proof}

\begin{corollary}\label{corollary-13-15}
Let $G$ be a connected subcubic multigraph not containing $G_3$ as a subgraph and different from $B_3$ and $G^*_5$. There exists a 3-edge-colorable subgraph of $G$ with at least $\frac{13}{15} |E(G)|$ edges. Moreover, this subgraph and its coloring can be found in polynomial time.
\end{corollary}

\begin{proof}
Suppose that the theorem is not true. Let $G$ be a counter-example with the least number of vertices.

It is easy to check that if every biconnected component of $G$ has a 3-edge-colorable subgraph with at least $\frac{13}{15}$ of its edges, then so does $G$. Thus, by Theorem \ref{theorem-13-15} we can assume that there exists a biconnected component $C$ of $G$ which is isomorphic to $B_3$ or $G^*_5$. Since $C$ is not the whole multigraph, there is an edge $vw$ with $v \in V(C)$ and $w\not\in V(C)$. 
If $C \cup vw$ is the whole multigraph, it does have a 3-edge-colorable subgraph with at least $\frac{13}{15} |E(G)|$ edges. Hence, $H := G[V \setminus (V(C) \cup \{w\})]$ is not empty. 

Notice that $vw$ is a bridge.  Since $C \cup \{vw\}$ has a 3-edge-colorable subgraph with at least $\frac{13}{15}$ of its edges, and $w$ is a cut-vertex, then -- by a similar reasoning as above -- $G[V(H) \cup \{w\}]$ does not have a 3-edge-colorable subgraph with at least $\frac{13}{15}$ of its edges. By minimality of $G$, $G[V(H) \cup \{w\}]$ is isomorphic to $B_3$ or $G^*_5$. However, then, $G$ is a cubic multigraph with 15 edges and it has a 3-edge-colorable subgraph with at least 13 edges; a contradiction.
\end{proof}

\begin{corollary}\label{corollary-7-9}
Every connected subcubic multigraph $G$ different from $G_3$ has a 3-edge-colorable subgraph with at least $\frac{7}{9} |E(G)|$ edges. Moreover, this subgraph and its coloring can be found in polynomial time.
\end{corollary}

\begin{proof}
Let $G$ be a connected multigraph different from $G_3$. We use induction on $|V(G)|$. First, assume that $G$ is also biconnected. If $G$ is isomorphic to $B_3$ or $G^*_5$, then it has a 3-edge-colorable subgraph with at least $\frac{6}{7} |E(G)| \geq \frac{7}{9} |E(G)|$ edges. Otherwise, from Theorem~\ref{theorem-13-15}, it has a 3-edge-colorable subgraph with at least $\frac{13}{15} |E(G)| \geq \frac{7}{9} |E(G)|$ edges. 

Now, let us assume that $G$ has a cut-vertex $v$. Since $G$ is subcubic, it has also a cut-edge $vw$. 
Let $C'$ and $C''$ be the connected components of $G-vw$. If both $C'$ and $C''$ have 3-edge-colorable subgraphs with at least $\frac{7}{9}$ of its edges, then so does $G$. Hence, by the induction hypothesis we can assume that at least one component (say, $C'$) is isomorphic to $G_3$.
If $C''$ is not isomorphic to $G_3$, then by the induction hypothesis we can color $\frac{7}9$ of the edges of $C''$. Next, we can color four out of the five edges of $C'\cup\{vw\}$ and thus we obtain a 3-edge-colorable subgraph of $G$ with more than $\frac{7}{9} |E(G)|$ edges.
In the remaining case the whole multigraph consists of two copies of $G_3$ with the degree 2 vertices connected by the edge $vw$. It has 9 edges and a 3-edge-colorable subgraph with 7 edges.
\end{proof}

\section{Large $\Delta$-edge-colorable subgraphs in simple graphs with maximum degree $\Delta$ from four to seven}

In this section by a graph we mean a simple graph. We prove the following theorem.

\begin{theorem}
\label{thm:main}
Let $G$ be a connected simple graph of maximum degree $\Delta\in\{4, 5, 6, 7\}$. Then $G$ has a $\Delta$-edge-colorable subgraph with at least 
\begin{enumerate}[a)]
 \item $\frac{5}{6}|E|$ edges when $\Delta=4$ and $G\ne K_5$,
 \item $\frac{23}{27}|E|$ edges when $\Delta=5$,
 \item $\frac{19}{22}|E|$ edges when $\Delta=6$ and $G\ne K_7$,
 \item $\frac{22}{25}|E|$ edges when $\Delta=7$.
\end{enumerate}
Moreover, the subgraph can be found in polynomial time.
\end{theorem}

We will work with partially colored graphs. A {\em partial $k$-coloring} of a graph $G=(V,E)$ is a function $\pi : E \rightarrow \{1, \ldots, k\} \cup \{\bot\}$ such that if two edges $e_1, e_2\in E$ are incident then $\pi(e_1) \ne \pi(e_2)$, or $\pi(e_1) = \bot$, or $\pi(e_2) = \bot$. We will call the pair $(G,\pi)$ a {\em colored graph}.
We say an edge $e$ is {\em uncolored} if $\pi(e)=\bot$; otherwise, we say that $e$ is {\em colored}. For a vertex $v$, ${\pi}(v)$ is the set of colors of edges incident with $v$, i.e.\ $\pi(v)=\{\pi(e)\ :\ \text{$e$ is incident with $v$}\}\setminus\{\bot\}$, while $\overline{\pi}(v)=\{1, \ldots, k\}\setminus \pi(v)$ is the set of free colors at $v$. 

Our plan for proving Theorem~\ref{thm:main} is the following. 
We introduce a notion of the potential function $\Psi$, which measures ``the quality'' of a partial $\Delta$-coloring $\pi$ of a given graph $G$. It turns out that if we are unable to improve the potential of a partial coloring $\pi$ then the pair $(G,\pi)$ exhibits certain structure. 
We are going to determine this structure in a series of lemmas so that we are able to show that $\pi$ has few uncolored edges. 
In the proofs of the structural lemmas we show that if the claim of the lemma does not hold, one can find in polynomial time a new coloring so that the potential increases. Hence, in order to find a partial coloring which satisfies the claimed lower bound on the number of colored edges it suffices to start with an empty coloring and then, as long as the claim of some of the structural lemmas does not hold, find a new coloring with improved potential, as described in the relevant proof.
Since, as we will see, the potential can be increased only polynomial number of times, the whole procedure works in polynomial time.

\subsection{The structure of maximum $\Delta$-edge-colorable subgraphs}
\label{sec:structure}

Let $G$ be an arbitrary connected graph and let $\Delta$ denote its maximum degree.
In this section we study the structure of a partial edge-coloring $\pi$ of $G$, such that the number of colored edges cannot be increased.
We defer the choice of our potential $\Psi$ until we show the full motivation for its definition. 
However, the potential $\Psi$ grows with the number of colored edges, so the structure of $(G, \pi)$ described in this section applies also when $\Psi$ cannot be increased.
Another reason for deferring its full description is that we prefer to state the claims of this section under weaker assumptions since we believe they might be useful in further research.

Let $a$ and $b$ be two distinct colors and $x$ and $y$ be two distinct vertices.
An {\em $(ab, xy)$-path} is a path $P=x_1x_2\ldots x_t$ for some $t>0$, such that:
\begin{itemize}
 \item $x=x_1$ and $y=x_t$,
 \item the edges of $P$ are colored alternately with $a$ and $b$, i.e.\ $\pi(x_{i}x_{i+1})\in\{a,b\}$ and if $\pi(x_{i}x_{i+1})=a$ and $\pi(x_{j}x_{j+1})=b$ then $i \not\equiv j \bmod 2$,
 \item $P$ is maximal, i.e.\ $|\f(x)\cap\{a,b\}|=|\f(y)\cap\{a,b\}|=1$.
\end{itemize}
We also say that $P$ is an alternating path, $(ab,\cdot)$-path, $(ab,x)$-path, $(\cdot, xy)$-path or $(a,xy)$-path.

The idea of alternating paths dates back to Kempe~\cite{kempe1879geographical} and his first attempts to prove the Four Color Theorem. 
The basic property of an alternating path $P$ is that we can recolor the graph along $P$ so that
all edges of $P$ colored with $a$ get color $b$ and vice versa.
Note that as a result, if $a$ (resp. $b$) was free in one end of the path $P$, say in $x$ then in $\f(x)$ the color $a$ is replaced by $b$ (resp. $b$ is replaced by $a$), and for every vertex $v\not\in\{ x,y\}$ the set of free colors $\f(v)$ stays the same. 
We will often use this operation, called {\em swapping}.

Let $V_{\bot} = \{v \in V\ :\ \f(v)\ne \emptyset\}$. In what follows, $\bot(G,\pi)=(V_{\bot},\pi^{-1}(\bot))$ is called {\em the graph of free edges}.
Every connected component of the graph $\bot(G,\pi)$ is called a {\em free component}. If a free component has only one vertex, it is called {\em trivial}.
The set of all nontrivial free components of colored graph $(G,\pi)$ is denoted by $\nfc(G,\pi)$. 

\begin{lemma}
\label{lem:distinct-free}
Let $G$ be a graph and let $\pi$ be a partial coloring of $G$ which maximizes the number of colored edges.
For any free component $Q$ of $(G,\pi)$ and for every two distinct vertices $v,w\in V(Q)$
\begin{enumerate}[$(i)$]
 \item $\f(v) \cap \f(w) = \emptyset$,
  \item for every $a\in \f(v)$, $b\in\f(w)$ there is an $(ab, vw)$-path. 
\end{enumerate}
\end{lemma}

\begin{proof}
First we prove $(i)$ and we use induction on the length $d$ of the shortest path $P$ in $Q$ from $v$ to $w$.
The proof is by contradiction, i.e.\ we show that if $\f(v)\cap\f(w)\ne\emptyset$ then one can increase the number of colored edges. 
If $d=1$ just color $vw$ with a color from $\f(v)\cap\f(w)$.
Now we consider $d>1$.
Assume there is a color $a \in \f(v) \cap \f(w)$.
Let $x$ be the second to last vertex on $P$, i.e. $xw\in E(P)$. 
Since $x$ is incident with an uncolored edge, there is a free color at $x$, say $b$.
Since we have already proved the claim for $d=1$, we infer that $a\ne b$ and $b\not\in\f(w)$.
Let $R$ be the $(ab,w)$-path. We swap $R$. If $x$ is not incident with $R$ then $b$ is free at both $x$ and $w$ and
we just color $xw$ with $b$ and we increase the number of colored edges; a contradiction.
If $x$ is incident with $R$ it means that $R$ is an $(ab, wx)$-path. Hence after swapping, $a\in\f(v)\cap \f(x)$.
Since $v$ and $x$ are at distance $d-1$ in $Q$ we get a contradiction with the induction hypothesis.

To see $(ii)$, just consider the $(ab,v)$-path and note that by $b\not\in\f(v)$ by $(i)$ so the path has length at least one. If this path does not end in $w$ we can swap it and get $b\in\f(v)\cap \f(w)$, contradicting $(i)$. Also, by $(i)$, we have $v\ne w$.
\end{proof}

For a free component $Q$, by $\f(Q)$ we denote the set of free colors at the vertices of $Q$, i.e.\ $\f(Q)=\bigcup_{v\in V(Q)}\f(v)$.

\begin{corollary}
\label{cor:num-uncolored-0}
Let $G$ be a graph and let $\pi$ be a partial coloring of $G$ which maximizes the number of colored edges.
For any free component $Q$ of $(G,\pi)$ we have $|\f(Q)|\ge 2|E(Q)|$.
\end{corollary}
\begin{proof}
We have $|\f(Q)| = \sum_{v\in V(Q)}|\f(v)| \ge \sum_{v\in V(Q)}\deg_Q(v) \ge  2|E(Q)|$, where the first equality follows from Lemma~\ref{lem:distinct-free}$(i)$. 
\end{proof}

Since $|\f(Q)| \le \Delta$ we immediately get the following.

\begin{corollary}
\label{cor:num-uncolored}
Let $G$ be a graph and let $\pi$ be a partial coloring of $G$ which maximizes the number of colored edges.
Every free component $Q$ of $(G,\pi)$ has at most $\lfloor\frac{\Delta}{2}\rfloor$ edges.\qed
\end{corollary}

Let $Q_1,Q_2$ be two distinct free components of $(G,\pi)$ and assume that for some pair of vertices $x\in V(Q_1)$ and $y\in V(Q_2)$, there is an edge $xy\in E$ such that $\pi(xy)\in\f(Q_1)$. Then we say that {\em $Q_1$  sees $Q_2$ with $xy$}, or shortly {\em $Q_1$ sees $Q_2$}.

\begin{lemma}
\label{lem:A}
Let $G$ be a graph and let $\pi$ be a partial coloring of $G$ which maximizes the number of colored edges.
If $Q_1,Q_2$ are two distinct free components of $(G,\pi)$ such that $Q_1$ sees $Q_2$ then
$\f(Q_1) \cap \f(Q_2) = \emptyset$.
\end{lemma}

\begin{proof}
Let $x\in V(Q_1)$, $y\in V(Q_2)$ be vertices such that $Q_1$ sees $Q_2$ with $xy$.
Denote $a=\pi(xy)$. Let $v$ be a vertex of $Q_1$ such that $a\in\f(v)$. The proof is by contradiction.

First assume $a\in \f(Q_2)$. Since $a\not\in \f(y)$ it follows that $|E(Q_2)|>0$ and, in particular, $y$ has a neighbor in $Q_2$, say $y'$. By Lemma~\ref{lem:distinct-free}$(i)$ there is exactly one vertex $z\in V(Q_2)$ such that $a\in\f(z)$. 
Now we use induction on the length $d$ of the shortest path $P$ in $Q_2$ from $y'$ to $z$.
If $d=0$, i.e. $z=y'$ we uncolor $xy$ and we color $yy'$ with $a$. As a result, the number of colored edges has not changed and we get a free component in which two vertices (namely, $v$ and $x$) share the same free color $a$, which is a contradiction with Lemma~\ref{lem:distinct-free}$(i)$.
Now assume $d>0$ and let $z'$ be the second to last vertex on $P$, i.e. $z'z\in E(P)$. 
Let $c$ be any color of $\f(z')$. Consider the $(ac,zz')$ path $R$ described in Lemma~\ref{lem:distinct-free}$(ii)$. 
If $R$ does not contain $xy$, we just swap $R$ (note that after the swapping we still have $a\in\f(Q_1)$) and proceed by induction hypothesis.
Otherwise let $R'$ be the maximal subpath of $R$ which starts in $z$ and does not contain $xy$. We uncolor $xy$, swap $R'$ and color $zz'$ with $c$. Again, the number of colored edges has not changed and we get a free component with two vertices (namely, $v$ and the endpoint of $xy$ which is not incident with $R'$) that share the same free color $a$.

Now assume that for some color $b\ne a$ we have $b\in\f(x')\cap\f(y')$ for some $x'\in V(Q_1)$ and $y'\in V(Q_2)$. If $x'\ne x$, choose any color $c \in \f(x)$ and swap the $(bc,x'x)$-path described in Lemma~\ref{lem:distinct-free}$(ii)$. 
We proceed analogously when $y'\ne y$. 
Hence we can assume that $b\in\f(x)\cap\f(y)$. 
Then we recolor $xy$ to $b$. 
As a result, $a\in \f(v)\cap\f(x)$ and $v$ and $x$ still belong to the same free component, which is a contradiction with Lemma~\ref{lem:distinct-free}$(i)$. 
\end{proof}

\ignore{
\begin{lemma}[(not used?)]
\label{lem:B}
Let $G$ be a graph and let $\pi$ be a partial coloring of $G$ which maximizes the number of colored edges.
Let $P$, $Q$ and $R$ be three distinct free components of $(G,\pi)$.
Assume that $|\f(P)|\ge 3$, and $P$ sees both $Q$ and $R$.
Then the sets $\f(P)$, $\f(Q)$ and $\f(R)$ are pairwise disjoint.  
\end{lemma}

\begin{proof}
Let $p_1\in V(P)$, $q\in V(Q)$ be the vertices such that $P$ sees $Q$ with $p_1q$.
Let $p_2\in V(P)$, $r\in V(R)$ be the vertices such that $P$ sees $R$ with $p_2r$.
Let $a$ be any color from $\f(P)\setminus\{\pi(p_1q),\pi(p_2r)\}$.

The proof is by contradiction: we assume that the sets $\f(P)$, $\f(Q)$ and $\f(R)$ are not pairwise disjoint.  
By Lemma~\ref{lem:A}, $\f(P)\cap f(Q)=\emptyset$ and $\f(P)\cap f(R)=\emptyset$. 
Hence there exists a color $b\in \f(Q)\cap f(R)$ such that $b\not\in\f(P)$ (in particular $b\not\in\{a,\pi(p_1q),\pi(p_2r)\}$).

Let $v$ be the vertex of $P$ such that $a\in\f(v)$. We swap the $(ab,v)$-path. Then $b\in\f(P)$ and $P$ still sees $Q$ and $R$. Moreover $b\in\f(Q)$ or $b\in\f(R)$. This is a contradiction with Lemma~\ref{lem:A}.
\end{proof}
}

\begin{lemma}
\label{lem:C}
Let $G$ be a graph and let $\pi$ be a partial coloring of $G$ which maximizes the number of colored edges.
Let $P$, $Q$ and $R$ be free components of $(G,\pi)$, $P\ne Q$ and $P\ne R$.
Assume that for some $x\in P$ and $y\in Q$ there is an edge $xy\in E(G)$ and for some $u\in P$ and $v\in R$ there is an edge $uv\in E(G)$, $xy\ne uv$.
If $\pi(xy)=\pi(uv)$ then there are no two distinct colors $a,b\in \f(P)$ such that $a\in \f(Q)$ and $b\in \f(R)$.
\end{lemma}

\begin{proof}
The proof is by contradiction.

Let $x'$ be the vertex of $P$ such that $a\in\f(x')$ and let $c$ be any color of $\f(x)$. 
By Lemma~\ref{lem:distinct-free}$(i)$, $a\ne c$.
Note that by Lemma~\ref{lem:A} we have $\pi(xy)\ne a$. In particular, $\pi(xy)=\pi(uv)\ne a,c$.
If $x'\ne x$ we swap the $(ac,xx')$ path described in Lemma~\ref{lem:distinct-free}$(ii)$. Note that the colors of $xy$ and $uv$ do not change.
Similarly, let $y'$ be the vertex of $Q$ such that $a\in\f(y')$ and let $d$ be any color of $\f(y)$. 
Again, $\pi(xy)=\pi(uv)\ne a,d$. If $y'\ne y$ we swap the $(ad,yy')$ path described in Lemma~\ref{lem:distinct-free}$(ii)$ and again this does not change the colors of $xy$ and $uv$. Observe also that the sets of free colors of $P$ and $Q$ have not changed.
Then we recolor $xy$ to $a$. After this operation, $\pi(uv)$ becomes free in $P$ and $P$ sees $R$ with $uv$. However, $b \in \f(P)\cap\f(R)$; a contradiction with Lemma~\ref{lem:A}.
\end{proof}

\begin{corollary}
\label{cor:C1}
Let $G$ be a graph and let $\pi$ be a partial coloring of $G$ which maximizes the number of colored edges.
 Let $Q$ be a free component of $(G,\pi)$ such that $\Delta-1 \le |\f(Q)| \le \Delta$. Then there are at most $\Delta-|\f(Q)|$ edges incident both with $Q$ and other nontrivial free components.
 Moreover, each such an edge is colored with a color from $\{1,\ldots,\Delta\}\setminus \f(Q)$.
\end{corollary}

\begin{proof}
 We can assume that $|\f(Q)|=\Delta-1$ for otherwise by Lemma~\ref{lem:A} there are no edges incident with $Q$ and other free components and the claim follows.
 We infer that there is exactly one color $c\not\in \f(Q)$.
 
 Assume to the contrary, that there are two edges $xy$ and $uv$ with the property described in the statement, with $x,u\in V(Q)$.
 Let $P$ and $R$ be the nontrivial free components such that $y\in V(P)$ and $v\in V(R)$, possibly $P=R$.
 Any nontrivial free component has at least two free colors by Lemma~\ref{lem:distinct-free}$(i)$, so in particular it has a color from $\f(Q)$, and hence by Lemma~\ref{lem:A} both $xy$ and $uv$ are colored with $c$ (this, in particular, proves the second part of the claim). 
 Then $c\not\in \f(P)\cup \f(R)$ for otherwise $P$ or $R$ sees $Q$; a contradiction with Lemma~\ref{lem:A}. It follows that both $\f(P)$ and $\f(R)$ are subsets of $\f(Q)$, both of cardinality at least 2, which is a contradiction with Lemma~\ref{lem:C}.
\end{proof}

Now we need another classical notion in the area of edge-colorings: the notion of a fan. We use a somewhat relaxed definition, due to Favrholdt, Stiebitz and Toft~\cite{stiebitz}, adapted to our setting of partially colored graphs.
Let $(G,\pi)$ be a partially edge-colored graph and let $xy$ be an uncolored edge of $G$.
An {\em $(x,y)$-fan} is a sequence of edges $F =
(xy_1, \ldots, xy_{\ell})$, where $y_1=y$ and for each $i=2,\ldots,\ell$ there is an index $\pred_F(i)<i$ such that the edge $xy_i$ is colored with a color $\pi(xy_i)\in \f(y_{\pred_F(i)})$. We say that a fan is {\em maximal} when it is not a proper prefix of another fan. The vertices $y_2,\ldots,y_{\ell}$ are called {\em ends} of $F$.
A proof of the following fact can be found in~\cite{stiebitz}; see Theorem 2.1: point (a) below appears explicitly while point (b) can be found in the proof.

\begin{lemma}[Favrholdt et al.~\cite{stiebitz}]
\label{lem:fan1}
 Let $F = (xy_1, \ldots, xy_{\ell})$ be a maximal fan in a partial $\Delta$-edge-coloring $(G,\pi)$ such that the number of colored edges cannot be increased. Then 
 \begin{enumerate}[(a)]
  \item if $1\le i < j \le \ell$ then $\f(y_i) \cap \f(y_j)=\emptyset$ and 
  \item $\{\pi(xy_2),\ldots,\pi(xy_{\ell})\} = \bigcup_{i=1}^{\ell}\f(y_i)$, 
 \end{enumerate}
\end{lemma}

For a fan $F = (xy_1, \ldots, xy_{\ell})$, if $\f(y_i)=\emptyset$ then we say $y_i$ is a {\em full vertex} and $xy_i$ is a {\em full edge}. 

\begin{corollary}
\label{cor:num-full}
Let $G$ be a graph and let $\pi$ be a partial coloring of $G$ which maximizes the number of colored edges.
 Any maximal $(x,y)$-fan $F$ in $(G,\pi)$ has at least $|\f(y)|$ full edges.
\end{corollary}

\begin{proof}
By Lemma~\ref{lem:fan1}, $\sum_{i=1}^{\ell}|\f(y_i)| = \ell - 1$. 
Let $f$ be the number of full edges of $F$. 
Then clearly, $\sum_{i=1}^{\ell}|\f(y_i)|\ge |\f(y)| + \ell - 1 - f$. Hence, $f \ge |\f(y)|$, as required. 
\end{proof}

Let $F=(xy_1,\ldots,xy_{\ell})$ be a fan in a partially colored graph $(G,\pi)$. Fix a vertex $y_i$, $i>1$. Define $\pred_F(1)=1$. Consider the following sequence of indices: $a_1=i$, and for every $j>1$, $a_j=\pred_F(a_{j-1})$.
Let $d=\min\{j\ :\ a_j = 1\}$. Consider the following recoloring procedure which transforms the coloring $\pi$ into a new coloring $\pi'$: begin with $\pi'=\pi$ and for every $j=2,\ldots,d$ put $\pi'(xy_{a_j})=\pi(xy_{a_{j-1}})$. Finally, uncolor $xy_i$. 
Note that $\pi'$ is a proper partial coloring with the same number of colored edges as $\pi$. 
This procedure is called {\em rotating $F$ at $y_i$}.

\begin{lemma}
\label{lem:fans-edge-disjoint}
Let $G$ be a graph and let $\pi$ be a partial coloring of $G$ which maximizes the number of colored edges.
Let $F_1$ be an $(x,y)$-fan in $(G,\pi)$ and let $F_2$ be an $(x,z)$-fan in $(G,\pi)$, for some $y\ne z$. Then $F_1$ and $F_2$ do not share an edge.
\end{lemma}

\begin{proof}
Let $F_1=(xy_1,\ldots,xy_{\ell})$ and $F_2=(xz_1,\ldots,xz_t)$.
 Assume $(i,j)$ is the lexicographically first pair of indices such that $y_i = z_j$.
 Let $c=\pi(xy_i)$.
 Rotate $F_1$ at $y_{\pred_{F_1}(i)}$ and rotate $F_2$ at $z_{\pred_{F_2}(j)}$.
 Note that by our choice of $(i,j)$ it is possible to perform both rotations.
 As a result, the number of colored edges does not change and we get a free component with color $c$ free at two vertices, namely $y_{\pred_{F_1}(i)}$ and $z_{\pred_{F_2}(j)}$; a contradiction with Lemma~\ref{lem:distinct-free}.
\end{proof}

\subsection{The structure of $\Psi_0$-maximal partial $\Delta$-edge-colorings}
\label{sec:struct-psi0}

Now we are ready to define a potential function $\Psi_0$ for a partial coloring $(G,\pi)$. 
Let $c$ be the number of colored edges, i.e.\ $c=|\pi^{-1}(\{1,\ldots,\Delta\})|$.
For every $i=1,\ldots,\floor{\Delta/2}$, let $n_i$ be the number of free components with $i$ edges.
Then
\[\Psi_0(G,\pi)=(c,n_{\floor{\Delta/2}},n_{\floor{\Delta/2}-1},\ldots,n_1). \]

We use the lexicographic order on tuples to compare values of $\Psi_0$. 
In what follows we study the structure of a partial $\Delta$-edge-coloring $\pi$ of a graph $G$ which is $\Psi_0$-maximal, i.e.\ there is no partial $\Delta$-edge-coloring $\pi'$ with $\Psi_0(G,\pi') > \Psi_0(G,\pi)$.
Note that the claims of the lemmas in Section~\ref{sec:structure} also hold for $(G,\pi)$. 

The intuition behind the choice of the potential $\Psi_0$ is as follows. Our goal is to find a partial coloring so that we can injectively assign many colored edges to every uncolored edge. 
As we will see, to maintain the injectiveness of the assignment, edges of a free component $Q$ are assigned only edges that are {\em close} to $Q$ (mostly edges incident with $Q$). In particular, if a colored edge is incident with two free components, we assign {\em half} of it to each of them. 
Assume $\Delta$ is even and consider a free component with $\Delta/2$ edges. Such a component will be called {\em maximal}. Observe that Corollary~\ref{cor:num-uncolored-0}  and Lemma~\ref{lem:A} imply that a colored edge is incident with at most one maximal component. 
Hence it seems that maximal components are good for us: they get assigned {\em the whole} incident edges, not just halves. This is why if we increase the number of maximal free components, our potential will increase, even if the number of colored edges stays the same. Our choice of $\Psi_0$ will also help when considering smaller components: for a smaller free component we will be able to argue that some (but not all) edges incident with it cannot be incident with another free component for otherwise by fan rotations we can ``merge'' the two components to form a bigger one.  
However, a rotation can increase the number of free components, and in particular it can decrease the potential. Hence we use rotations only for very special fans. 
Consider an $(x,y)$-fan $F$ and let $Q$ be the free component that contains $xy$. We say that $F$ is {\em stable}, if $Q-xy$ has no edges or $Q-xy$ has exactly one nontrivial (i.e.\ with at least one edge) connected component and this component contains $x$. (Note that even if every $(x,y)$-fan is stable it does not mean that a $(y,x)$-fan is stable).

\begin{proposition}
\label{prop:rotate}
Rotating a stable fan does not decrease $\Psi_0$. \qed 
\end{proposition}

\begin{lemma}
 \label{lem:usually-fans-do-not-meet}
  Let $G$ be a graph and let $\pi$ be a partial coloring of $G$ which maximizes the potential $\Psi_0$.
 Let $P$ and $Q$ be two distinct free components of $(G,\pi)$ and let $xy \in E(P)$, $zu \in E(Q)$.
 Assume $F_1=(xy_1,\ldots,xy_{\ell})$ is a stable $(x,y)$-fan.
 Let $F_2=(zu_1,\ldots,zu_t)$ be a $(z,u)$-fan.
 If $|E(Q)|\le|E(P)|$ or $F_2$ is stable then the ends of $F_1$ and $F_2$ are distinct, i.e.\ for every $i=1,\ldots,\ell$ and $j=1,\ldots,t$ we have $y_i \ne u_j$.
\end{lemma}

\begin{proof}
  Assume $(i,j)$ is the lexicographically first pair of indices such that $y_i = u_j$.
 Then we rotate $F_1$ at $y_i$ and we rotate $F_2$ at $u_j$ (note that because of the choice of $i$ and $j$, the free colors at $u_1,\ldots,u_j$ do not change during  the rotation of $F_1$ so the rotation of $F_2$ is still possible). 
 In the graph $\bot(G,\pi)$ it corresponds to removing edges $xy$ and $zu$ and adding edges $xy_i$ and $zu_j$ (note that $zu_j=zy_i$). Both when $|E(Q)|\le|E(P)|$ and when $F_2$ is stable the potential $\Psi_0$ increases (we get a new component of size at least $|E(P)|+1$ in the former case and of size exactly $|E(P)|+|E(Q)|$ in the latter case); a contradiction.
\end{proof}
 
Let $Q$ be a free component.
Then $S_1(Q)$ is the set of all vertices $v$ such that for some edge $xy\in E(Q)$ there is a stable $(x,y)$-fan which contains $xv$ as a full edge.
For any $v\in S_1(Q)$ the stable fan from the definition above is denoted by $F(v)$; if there are many such fans then we choose an arbitrary one as $F(v)$.
We also define $S(Q)=V(Q)\cup S_1(Q)$.
Note that by Lemma~\ref{lem:usually-fans-do-not-meet} for two distinct free components $Q$ and $R$ the sets $S(Q)$ and $S(R)$ are disjoint.
For two sets $A$ and $B$, any edge $ab$ with $a\in A$ and $b\in B$  will be called an {\em $AB$-edge}.

\begin{lemma}
 \label{lem:QR-edges-matching}
 Let $G$ be a graph  of maximum degree $\Delta$ and let $\pi$ be a partial coloring of $G$ which maximizes the potential $\Psi_0$.
Assume $\Delta$ is odd and let $Q$ be a free component of $(G,\pi)$ such that $|E(Q)|=(\Delta-1)/2$ and $|\f(Q)| = \Delta - 1$. 
Let $R$ be a free component, $R\ne Q$.
Then the set of all $S(Q)S(R)$-edges is a matching.
\end{lemma}

\begin{proof}
Let $c$ be the only color in $\{1,\ldots,\Delta\}\setminus\f(Q)$.
Consider an arbitrary $S(Q)S(R)$-edge $vw$, $v\in S(Q)$. 
We can assume that $v\in V(Q)$ for otherwise we rotate the stable fan $F(v)$ at $v$; note that then the component which replaces $Q$ has also $(\Delta-1)/2$ edges so by Corollary~\ref{cor:num-uncolored-0} it has at least $\Delta-1$ free colors.
Then $\pi(vw)=c$, because if $w\in V(R)$ this follows from Corollary~\ref{cor:C1} and otherwise, i.e.\ when $w\in S_1(R)$, we can rotate the fan $F(w)$ at $w$ and get $w\in V(R)$. (Note that rotating both $F(v)$ and $F(w)$ is possible because they are disjoint by Lemma~\ref{lem:usually-fans-do-not-meet}.) We have just proved that an arbitrary $S(Q)S(R)$-edge is colored by $c$, so the claim follows.
\end{proof}

\begin{lemma}
 \label{lem:QR-edges}
 Let $G$ be a graph  of maximum degree $\Delta$ and let $\pi$ be a partial coloring of $G$ which maximizes the potential $\Psi_0$.
Assume $\Delta$ is odd and let $Q$ be a free component of $(G,\pi)$ such that $|E(Q)|=(\Delta-1)/2$ and $|\f(Q)| = \Delta - 1$. 
Let $R$ be a free component, $R\ne Q$.
If there are at least two $V(Q)S(R)$-edges and at least two $S(Q)V(R)$-edges then
there is no $V(Q)V(R)$-edge.
\end{lemma}

\begin{proof}
In this proof we use the following definition.
Let $v_1\in S_1(P)$ and $v_2\in V(P)$ for some free component $P$.
We say that $v_1$ is {\em safe for $v_2$} if after rotating $F(v_1)$ at $v_1$ the vertices $v_1$ and $v_2$ are in the same free component.

Now we proceed with the proof. Assume on the contrary that there is an edge $qr$ such that $q\in V(Q)$ and $r\in V(R)$.
Let $q'r'$ be another $V(Q)S(R)$-edge, $q'\in V(Q)$, and let $q''r''$ be another $S(Q)V(R)$-edge, $r''\in V(R)$; both edges exist by our assumption.
Note that $r' \in S_1(R)$ and $q'' \in S_1(Q)$ for otherwise we get a contradiction with Corollary~\ref{cor:C1}, so in particular $q'r'\ne q''r''$.
By Lemma~\ref{lem:QR-edges-matching} we see that $q$, $q'$ and $q''$ are pairwise distinct, and so are $r$, $r'$ and $r''$.

If $r'$ is safe for $r$ then we rotate $F(r')$ at $r'$ and we get a new component $R'$ with two $V(Q)V(R')$-edges; a contradiction with Corollary~\ref{cor:C1}.

If $q''$ is safe for $q$ then we rotate $F(q'')$ at $q''$ and we get a new component $Q'$.
Since $|E(Q')|=|E(Q)|$ by Corollary~\ref{cor:num-uncolored-0} we have $|\f(Q')|\ge\Delta-1$.
However, there are two $V(Q')V(R)$-edges; a contradiction with Corollary~\ref{cor:C1}.

Now assume that $r'$ is not safe for $r$ and $q''$ is not safe for $q$.
Observe that any vertex $v\in S_1(P)$ can be not safe for at most one vertex, namely if $F(v)$ is a $(x,y)$-fan then $v$ can be not safe only for $y$.
Hence $r'$ is safe for $r''$ and $q''$ is safe for $q'$.
We rotate both $F(r')$ at $r'$ and $F(q'')$ at $q''$. As a result we get two new components $Q'$ and $R'$ where $q'r'$ and $q''r''$ are  $V(Q')V(R')$-edges.
By the same argument as before, $|\f(Q')|\ge\Delta-1$ so we get a contradiction with Corollary~\ref{cor:C1}.
\end{proof}
 
\subsection{The structure of a $\Psi$-maximal partial $\Delta$-edge-coloring}

Now we define our final potential function $\Psi$ for a partial coloring $(G,\pi)$.
Let $\#_c$ be the number of cycles in all free components. 
Recall that $\nfc(G,\pi)$ denotes the set of nontrivial (i.e., with at least one edge) free components of $(G,\pi)$.
Then 
\[\Psi(G,\pi)=(\Psi_0(G,\pi), \#_c, \Delta|V| - \sum_{Q\in\nfc(G,\pi)}|\f(Q)|). \]

Again assume that $(G,\pi)$ maximizes $\Psi$.
Note that all the  results from Sections~\ref{sec:structure} and~\ref{sec:struct-psi0} apply. 

\begin{lemma}
 \label{lem:cycles}
 Let $G$ be a graph and let $\pi$ be a partial coloring of $G$ which maximizes the potential $\Psi$.
 Let $F_1=(xy_1,\ldots,xy_{\ell})$ be a stable $(x,y)$-fan and $F_2=(zu_1,\ldots,zu_t)$ be a stable $(z,u)$-fan, where $xy$ and $zu$ are distinct edges of the same free component $Q$ of $(G,\pi)$. If $Q$ is a tree, then the ends of $F_1$ and $F_2$ are distinct, i.e.\ for every $i=1,\ldots,\ell$ and $j=1,\ldots,t$ we have $y_i \ne u_j$.
\end{lemma}

\begin{proof}
 Assume $y_i=u_j$ for some $i=1,\ldots,\ell$ and $j=1,\ldots,t$. 
 Since $F_1$ and $F_2$ are stable, $y$ and $u$ are leaves of $Q$. 
 Hence if $xy$ and $zu$ are incident then $x=z$ and the claim follows from Lemma~\ref{lem:fans-edge-disjoint}.
 Otherwise we perform the two rotations described in the proof of Lemma~\ref{lem:usually-fans-do-not-meet}.
 As a result we get a new component $Q'=Q-\{xy,zu\}\cup\{xy_i,y_iz\}$.
 Then not only $\Psi_0$ does not decrease but also $\#_c$ increases, so $\Psi$ increases; a contradiction.
\end{proof}

\begin{proposition}
 \label{prop:free-colors}
 Let $G$ be a graph and let $\pi$ be a partial coloring of $G$ which maximizes the potential $\Psi$.
 Let $Q$ be a free component of $(G,\pi)$ and let $xy$ be an edge of $Q$.
 Let $F$ be a stable $(x,y)$-fan and let $Q'$ be the free component that replaces $Q$ after rotating $F$.
 Then $|\f(Q')|\ge|\f(Q)|$.
\end{proposition}

\begin{proof}
Assume $|\f(Q')|<|\f(Q)|$. By Proposition~\ref{prop:rotate}, $\Psi_0$ does not decrease.
Let $F=(xy_1,\ldots,xy_{\ell})$ and assume $F$ is rotated at $y_i$. Assume $xy$ belongs to a cycle in $Q$. Then $V(Q')=V(Q)\cup\{y_{i}\}$.
Since the number of free colors in every vertex from $\{y_2,\ldots,y_{i-1}\}$ does not change after the rotation, the only vertex for which which the number of free colors decreases is $y$, but it stays in the component. Hence by Lemma~\ref{lem:distinct-free} we have $|\f(Q')|\ge|\f(Q)|$, a contradiction.
Since $xy$ does not belong to a cycle in $Q$ we infer that $\#_c$ does not decrease. 
Observe that $y_2,\ldots,y_{i}$ do not belong to nontrivial free components different from $Q$, for otherwise in the process of rotating $F$ we can merge two components and $\Psi_0$ increases. Hence, after the rotation all nontrivial components different from $Q'$ do not change their free colors. Then $|\f(Q')|<|\f(Q)|$ implies that $\sum_{R\in\nfc(G,\pi)}|\f(R)|$ decreases, so $\Psi$ increases; a contradiction.
\end{proof}

\begin{lemma}
 \label{lem:full-component-no-QR-edges}
 Let $G$ be a graph of maximum degree $\Delta$ and let $\pi$ be a partial coloring of $G$ which maximizes the potential $\Psi$.
Let $Q$ be a free component of $(G,\pi)$ such that $|\f(Q)|=\Delta$. 
Then for any other free component $R$ there are no $S(Q)S(R)$-edges.
\end{lemma}

\begin{proof}
The proof is by contradiction. Assume there is an edge $uv$, such that $u\in S(Q)$ and $v \in S(R)$.
First assume that $v\in S_1(R)$.
Then we rotate the fan $F(v)$ at $v$. Note that the number of colored edges does not change.
Note that by Lemma~\ref{lem:usually-fans-do-not-meet} rotating $F(v)$ does not affect stable fans of $Q$, so in particular $S_1(Q)$ does not change after the rotation.
Hence we can assume that $v\in V(R)$.
Now assume that $u\in S_1(Q)$.
Then we rotate $F(u)$ at $u$; again the number of colored edges does not change and moreover the new free component also has $\Delta$ free colors by Proposition~\ref{prop:free-colors}.
Hence we can assume that $u\in V(Q)$, i.e.\ $uv$ is incident with both $Q$ and $R$.
Since the number of colored edges is maximal this is a contradiction with Corollary~\ref{cor:C1}.
\end{proof}

\subsection{Bounding the number of uncolored edges}

In this section we assume that $(G,\pi)$ is a partially colored graph such that $\pi$ maximizes the potential $\Psi$ and our goal is to give a bound on the number of uncolored edges. Here is our plan: We put a {\em charge}, equal to 1 to every colored edge of graph $G$.
Next, every colored edge sends its charge to its endpoints following carefully selected rules.
Finally, we assign disjoint sets of vertices to nontrivial free components. 
Then, we show a lower bound on the total charge at vertices assigned to a nontrivial free component divided by the number of edges in this component. 
This gives the desired bound.
Let us be more precise now. 
The lemma below will be used in describing the sets of vertices assigned to free components.

\begin{lemma}
\label{lem:A_1}
 Assume $4\le\Delta\le 7$.
 For every free component $Q$ there is a set $A_1(Q)\subseteq S_1(Q)$ such that
 \begin{enumerate}[$(i)$]
  \item if $Q$ is a tree, $z_1, z_2 \in A_1(Q)$, $F(z_1)$ is an $(x_1,y_1)$-fan and $F(z_2)$ is an $(x_2,y_2)$-fan then $\{x_1,y_1\}\ne \{x_2,y_2\}$,
  \item if $|E(Q)|\le 2$ then $|A_1(Q)|=|E(Q)|$,
  \item if $|E(Q)|=3$ then $|A_1(Q)|=2$ if $Q$ is a tree and $|A_1(Q)|=3$ otherwise.
 \end{enumerate}
\end{lemma}

\begin{proof}
First assume $|E(Q)|=1$ and let $E(Q)=\{xy\}$. Pick any maximal $(x,y)$-fan $F$. Then $F$ is stable and by Corollary~\ref{cor:num-full} fan $F$ has at least one full edge $xz$. We put $A_1(Q)=\{z\}$. 

Now assume $|E(Q)|\ge 2$ and $Q$ is a tree.
Consider an arbitrary leaf $\ell$ of $Q$ and let $x\ell$ be the edge of $Q$ incident with $\ell$.
Pick any maximal $(x,\ell)$-fan $F_{\ell}$. Since $\ell$ is a leaf $F_{\ell}$ is stable.
By Corollary~\ref{cor:num-full}, $F_{\ell}$ has at least $|\f(\ell)|\ge 1$ full edges.
Pick any such edge $x v_\ell$.
Since $|E(Q)|\ge 2$ there are at least two leaves.
We pick an arbitrary pair of leaves $\ell_1, \ell_2$ and we put $A_1(Q)=\{v_{\ell_1},v_{\ell_2}\}$.
By Lemma~\ref{lem:cycles} the fans $F_{\ell_1}$ and $F_{\ell_2}$ are disjoint (note that we can apply the lemma since $\ell_1$ and $\ell_2$ are not the endpoints of the same edge), so $|A_1(Q)|=2$.

Finally assume $|E(Q)|=3$ and $Q$ is a cycle.
Pick any vertex $v\in V(Q)$.
Observe that for any $w\in V(Q)$ we have $|\f(w)|\ge 2$.
Hence, by Corollary~\ref{cor:num-full} and Lemma~\ref{lem:fans-edge-disjoint} there are at least 4 full fan edges incident with $v$. 
Moreover, since $Q$ is a cycle, for any $xy\in E(Q)$ all $(x,y)$-fans are stable.
Let $vu_1$, $vu_2$, $vu_3$ be three of the at least four full fan edges incident with $v$.
We put $A_1(Q)=\{u_1,u_2,u_3\}$.
\end{proof}

For every nontrivial free component $Q$ the set of vertices assigned to $Q$ is defined as $A(Q)=V(Q)\cup A_1(Q)$.
Note that $A(Q)\subseteq S(Q)$.
It follows that for any two distinct free components $P$ and $Q$ the sets $A(P)$ and $A(Q)$ are disjoint, since $S(P)$ and $S(Q)$ are disjoint. 
Observe also that some vertices of $G$ may not be assigned to any of the free components. 
Let us denote $A_0 (Q) = V(Q)$, $A = \bigcup_{Q\in \nfc(G,\pi)} A(Q)$, $A_0 = \bigcup_{Q\in \nfc(G,\pi)} A_0(Q)$, and $A_1 = \bigcup_{Q\in \nfc(G,\pi)} A_1(Q)$. 

Our rules for moving the charge are the following.
Let $xy$ be an arbitrary colored edge. By symmetry we can assume that if one of its endpoints is in $A_0$ then $x \in A_0$.
\begin{enumerate}[(R1)]
 \item $xy$ divides its charge equally between its endpoints in $A$, i.e.\ it sends $\frac{1}{|\{x,y\}\cap A|}$ to each of its endpoints from $A$, unless (R2) applies.
 \item If $x\in A(P)$, $y \in A_1(Q)$ for two distinct free components $P$ and $Q$ such that $|E(P)|\ge 2$ and $|E(Q)|=1$, then $xy$ sends $(1-\epsilon_\Delta)$ to $x$ and $\epsilon_\Delta$ to $y$, where
\[ \epsilon_4 = \tfrac{1}{2}, \quad \epsilon_5 = \tfrac{1}{4}, \quad \epsilon_6 = \tfrac{1}{12}, \quad \epsilon_7 = \tfrac{3}{28}.\]
\end{enumerate}

Let $\ch(v)$ denote the amount of charge received by a vertex $v$. For a set $S\subseteq V$ we denote $\ch(S)=\sum_{v\in S}\ch(v)$. The disjointness of the sets $A(Q)$ immediately gives the following.

\begin{proposition}
 \[\gamma_{\Delta}(G) \ge \min_{Q\in\nfc(G,\pi)}\frac{\ch(A(Q))}{\ch(A(Q))+|E(Q)|}.\]
\end{proposition}

In what follows we give lower bounds for the ratio $\frac{\ch(A(Q))}{\ch(A(Q))+|E(Q)|}$ for $\Delta=4,\ldots,7$ and $|E(Q)|=1,\ldots,\floor{\Delta/2}$, which is sufficient by Corollary~\ref{cor:num-uncolored}. We begin with some simple cases.

\begin{lemma}
\label{lem:incident}
 Let $e$ be a colored edge incident with a free component $Q$. Then the charge $e$ sends to $A(Q)$ is
 \begin{enumerate}[$(i)$]
  \item at least $\frac{1}2$,
  \item at least $1-\epsilon_\Delta$ if $e$ is a full edge of a non-stable fan,
  \item 1 if $e$ is a full edge of a stable fan.
 \end{enumerate}
\end{lemma}

\begin{proof}
The discharging rules easily imply $(i)$.
Let $e=vw$ for $v\in V(Q)$ and $w\not\in V(Q)$.
If $e$ is a full edge, then $\f(w)=\emptyset$, so $w \not\in A_0$ and hence the rules imply $(ii)$.
Finally, if $e$ is a full edge of a stable fan then by Lemma~\ref{lem:usually-fans-do-not-meet} there is no free component $P\ne Q$ such that $w\in A_1(P)$.
It follows that if $w \in A$ then $w\in A_1(Q)$, so by (R1) $e$ sends 1 to $A(Q)$ and $(iii)$ follows.
\end{proof}

\begin{lemma}
\label{lem:bound-incident-1-edge}
 Let $Q$ be a free component consisting of exactly one edge.
 Then, the edges incident with $Q$ send the charge of at least $\Delta$ to $A(Q)$.
\end{lemma}

\begin{proof}
 By Lemma~\ref{lem:incident} every colored edge incident with $Q$ sends at least $1/2$ to $A(Q)$. Since for every vertex $v\in V(Q)$ there are exactly $\Delta-\f(v)$ such edges, they send at least $\frac{1}2\sum_{v\in Q}(\Delta-|\f(v)|)$ to $A(Q)$.
 
 Let $E(Q)=\{xy\}$. Then we choose a maximal $(x,y)$-fan $F_1$ and a maximal $(y,x)$-fan $F_2$. 
 Note that both $F_1$ and $F_2$ are stable, since $|E(Q)|=1$.
 The fan $F_1$ (resp.\ $F_2$) has at least $|\f(y)|$ (resp.\ $|\f(x)|$) full edges by Corollary~\ref{cor:num-full}. Hence there are at least $\sum_{v\in Q}|\f(v)|$ full fan edges incident with $Q$ and by Lemma~\ref{lem:incident} each of them sends $1$ to $A(Q)$. It follows that the total charge $A(Q)$ receives from the incident edges is at least $\frac{1}2\sum_{v\in Q}(\Delta-\f(v))+\frac{1}{2}\sum_{v\in Q}|\f(v)| = \Delta$. 
 \end{proof}

\begin{proposition}
\label{prop:bound-full-end}
Let $F$ be a stable $(x,y)$-fan and let $xz$ be a full edge of $F$.
If $z\in A_1(Q)$ for some free component $Q$, then the charge received by $z$ from edges not incident with $Q$ is at least $\eta_\Delta(\Delta-|V(Q)|)$, where
\[\eta_\Delta = \begin{cases}
                 \frac{1}{2}  & \text{when $|E(Q)|\ge 2$} \\
                 \epsilon_\Delta & \text{when $|E(Q)|=1$.}
           \end{cases}
\]
\end{proposition}

\begin{corollary}
\label{cor:bound-1-edge}
Let $Q$ be a one-edge free component of $(G,\pi)$.
Then,
\[\frac{\ch(A(Q))}{\ch(A(Q))+|E(Q)|} \ge \begin{cases}
                                            \frac{5}6 & \text{when $\Delta=4$,}\\ 
                                            \frac{23}{27} & \text{when $\Delta=5$,}\\ 
                                            \frac{19}{22} & \text{when $\Delta=6$,}\\ 
                                            \frac{211}{239}>\frac{22}{25} & \text{when $\Delta=7$.} 
                                         \end{cases}
\]
\end{corollary}

\begin{proof}
By Lemma~\ref{lem:A_1} we have $|A_1(Q)|=1$. Hence, by Lemma~\ref{lem:bound-incident-1-edge} and Proposition~\ref{prop:bound-full-end} we have $\ch(A(Q)) \ge \Delta + \eta_\Delta(\Delta-2)$, which is equal to $5$, $\frac{23}{4}$, $\frac{19}3$ and $\frac{211}{28}$ when $\Delta=4,5,6,7$, respectively. The claim follows.
\end{proof}

From our charge moving rules and Lemma~\ref{lem:full-component-no-QR-edges} we immediately get the following.

\begin{proposition}
\label{prop:bound-full-component}
For every free component $Q$ such that $|\f(Q)|=\Delta$ every edge incident with $A(Q)$ sends 1 to $A(Q)$. \qed
 \end{proposition}

\begin{lemma}
\label{lem:degrees-full-component}
 Let $Q$ be a free component such that $|E(Q)|=\floor{\Delta/2}$.
 Then, $A(Q)$ contains exactly $|\f(Q)|-2\floor{\Delta/2}$ vertices of degree $\Delta-1$ in $G$ and all the remaining vertices of $A(Q)$ are of degree $\Delta$.
 \end{lemma}

\begin{proof}
Clearly, for every $v\in V(Q)$ we have $|\f(v)| = \Delta - |\pi(v)| = \Delta - (\deg_G(v) - \deg_Q(v))$.
By Lemma~\ref{lem:A}, $|\f(Q)|=\sum_{v\in V(Q)}\f(v)$, so 
\[|\f(Q)| = \sum_{v\in V(Q)}(\Delta-\deg_G(v)) + 2|E(Q)|.\]
By plugging in our assumptions and rearranging the formula, we get 
\[\sum_{v\in V(Q)}(\Delta-\deg_G(v)) = |\f(Q)| -  2\floor{\Delta/2}.\]
By Corollary~\ref{cor:num-uncolored-0} we have $|\f(Q)|\ge 2|E(Q)|\ge \Delta-1$.
Hence, $|\f(Q)| -  2\floor{\Delta/2} \le 1$.
It follows that $Q$ has exactly $|\f(Q)|-2\floor{\Delta/2}$ vertices of degree $\Delta-1$ in $G$ and all the remaining vertices of $Q$ are of degree $\Delta$.
Moreover, since the vertices of $A_1(Q)$ are ends of {\em full} fan edges, each of them is of degree $\Delta$.
The claim follows.
\end{proof}

\begin{corollary}
\label{cor:bound-full-even-component}
 When $\Delta$ is even, for every free component $Q$ such that $|E(Q)|=\Delta/2$, 
 \[\ch(A(Q)) \ge \Delta|A(Q)| - |E(G[A(Q)])| - |E(Q)|.\]
 \end{corollary}
 
 \begin{proof}
  By Corollary~\ref{cor:num-uncolored-0} we have $|\f(Q)|=\Delta$. Hence by Lemma~\ref{lem:degrees-full-component} there are exactly  $\Delta|A(Q)| - |E(G[A(Q)])|$ edges incident with $A(Q)$, and $|E(Q)|$ of them are uncolored. This, together with Proposition~\ref{prop:bound-full-component}, gives the claim.
\end{proof}

Now we are very close to establishing our bound for $\Delta=4$. We will need just one more auxiliary claim (Lemma~\ref{lem:K_(k+1)-e} below).

\begin{lemma}[Folklore, see e.g.~\cite{bryant}]
\label{lem:odd-clique}
 For every odd $k$, the clique $K_{k+1}$ is $k$-edge colorable.
\end{lemma}

\begin{lemma}
\label{lem:even-clique}
 For every even $k$ we have $c(K_{k+1})=k^2/2$.
 Moreover, there is a partial $k$-edge-coloring $\pi$ of $K_{k+1}$ with $k^2/2$ colored edges such that the uncolored edges form a matching, and for each pair of distinct vertices $x$ and $y$, $\f(x)\ne\f(y)$.
\end{lemma}

\begin{proof}
Since every color class covers at most $\floor{(k+1)/2}=k/2$ edges, we have $c(K_{k+1})\le k^2/2$.

Now we show that $k^2/2$ edges of $K_{k+1}$ can be colored with $k$ colors.
Begin by a $(k+1)$-edge-coloring of $K_{k+2}$, which exists by Lemma~\ref{lem:odd-clique}. Remove one vertex to get a $(k+1)$-colored $K_{k+1}$. 
Uncolor the edges colored with the color $k+1$. There are at most $k/2$ of them, so the the number of colored edges is at least ${k+1\choose 2} - k/2 = k^2/2$.
The coloring satisfies the desired property because in the $(k+1)$-coloring of $K_{k+2}$ every vertex, including the removed one, is incident with all $k+1$ colors.
\end{proof}

Let $\G^\Delta_d$ be the family of all simple graphs which (i) have at least one edge, (ii) are of maximum degree at most $\Delta$ and (iii) such that any subset of vertices of size $(\Delta+1)$ induces a subgraph with at most ${\Delta+1\choose 2} - d$ edges.

\begin{lemma}
\label{lem:K_(k+1)-e}
Assume $\Delta\ge 4$ and $\Delta$ is even.
If for every graph $G\in\G^\Delta_2$ we have $\gamma_\Delta(G)\ge \alpha$ for some constant $\alpha\in[0,1]$, then
for every graph $G\in\G^\Delta_1$ we have $\gamma_\Delta(G)\ge \min\{\alpha,\frac{\Delta^2}{\Delta^2+\Delta-2}\}$. 
\end{lemma}

\begin{proof}
Let $G$ be an arbitrary graph from $G\in\G^\Delta_1$. 
We use induction on $|V(G)|$.
For the base case when $|V(G)|=2$, i.e.\ $G$ consists of a single edge, $\gamma_\Delta(G)=1$ so the claim follows.
Let $|V(G)|>3$. 
We can assume that $V(G)$ contains a subset $S$ of size $(\Delta+1)$ such that $|G[S]| = {\Delta+1 \choose 2}-1$ for otherwise the claim follows from the assumed property of $\G^\Delta_2$.
If there are no edges leaving $S$, then we just color $G-S$ inductively and we color $\Delta^2/2$ edges of $G[S]$ according to Lemma~\ref{lem:even-clique}.
Then $\gamma_\Delta(G)\ge\min\{\gamma_\Delta(G-S),(\Delta^2/2)/({\Delta+1\choose 2}-1)\}=\min\{\gamma_\Delta(G-S),\frac{\Delta^2}{\Delta^2+\Delta-2}\}\ge \min\{\alpha,\frac{\Delta^2}{\Delta^2+\Delta-2}\}$.
We can also assume that $G$ has no cutvertex for otherwise it is easy to get the claim from the induction hypothesis.
It follows that there are exactly two edges leaving $S$, say $xx'$ and $yy'$, with $x,y\in S$, and $x,x',y,y'$  distinct.
Then we remove $S$ and add a new vertex $q$ and two new edges $x'q$, $y'q$. Denote the resulting graph by $G'$.
Find the partial coloring of $G'$ corresponding to the largest $\Delta$-colorable subgraph of $G'$.
Then in the partially colored $G'$ we remove $q$ and put back the set $S$ with incident edges.
Color $xx'$ and $yy'$ with the colors of $x'q$ and $y'q$, respectively (and if one of the edges $x'q$, $y'q$ is uncolored, then the corresponding edge is also uncolored; note that $x'q$ and $y'q$ do not get the same color).
By Lemma~\ref{lem:even-clique} we can color $\Delta^2/2$ edges of $G[S]$ so that the edges of $G[S]$ incident with $x$ do not get the color of $xx'$ and the edges of $G[S]$ incident with $y$ do not get the color of $yy'$. Then again $\gamma_\Delta(G)\ge\min\{\gamma_\Delta(G'),(\Delta^2/2)/({\Delta+1\choose 2}-1)\}\ge \min\{\alpha,\frac{\Delta^2}{\Delta^2+\Delta-2}\}$.
\end{proof}

\begin{lemma}
\label{lem:bound-4-colors-2-edges}
Let $\Delta=4$ and let $Q$ be a two-edge free component of $(G,\pi)$.
If $G\in \G^\Delta_2$ then $\frac{\ch(A(Q))}{\ch(A(Q))+|E(Q)|} \ge \frac{5}{6}$. 
\end{lemma}

\begin{proof}
By Lemma~\ref{lem:A_1}, $|A_1(Q)|= 2$. 
By Corollary~\ref{cor:bound-full-even-component} we have $\ch(A(Q))\ge 4\cdot 5 - ({5 \choose 2} -2) - 2 = 10$ and we get $\ch(A(Q)) / (\ch(A(Q))+|E(Q)|) \ge \frac{5}{6}$, as required.
\end{proof}

\begin{corollary}
\label{cor:main-4}
Every connected simple graph $G$ of maximum degree $4$ has a $4$-edge-colorable subgraph with at least $\frac{5}{6}|E|$ edges, unless $G = K_5$.
\end{corollary}

\begin{proof}
 By Corollary~\ref{cor:num-uncolored} every free component of a partially 4-edge-colored graph which maximizes the potential $\Psi$ has at most two edges.
 Hence, by Corollary~\ref{cor:bound-1-edge} and Lemma~\ref{lem:bound-4-colors-2-edges} for every $G\in \G^4_2$ we have $\gamma_4(G)\ge\frac{5}6$.
 By Lemma~\ref{lem:K_(k+1)-e} the same bound holds also for graphs in $\G^4_1$, which is equivalent to our claim.
\end{proof}

\begin{lemma}
\label{lem:bound-almost-full-component}
Assume $\Delta\in \{5,7\}$ and let $Q$ be a free component such that $|E(Q)|=(\Delta-1)/2$, $|\f(Q)| = \Delta-1$. 
Let $D$ be the set of colored edges incident with $A(Q)$.
Then, the charge sent from $D$ to $A(Q)$ is at least $|D|-\frac{|A(Q)|-2}2$.
 \end{lemma}

\begin{proof}
%
Call an edge $e\in D$ {\em bad} if it sends less than 1 to $A(Q)$.
Note that every bad edge sends either $\frac{1}2$ or $1-\epsilon_\Delta\ge\frac{1}2$ to $A(Q)$.
Hence in what follows we assume that there are at least $|A(Q)|-1$ bad edges, for otherwise 
we get the claim immediately.

Clearly, every bad edge has only one endpoint in $A(Q)$ and the other endpoint is in $A(P)$ for some $P\ne Q$. We prove the following two auxiliary claims:

\noindent {\bf Claim 1:} There is a free component $P\ne Q$ such that every bad edge has an endpoint in $A(P)$.

\noindent {\em Proof of Claim 1.} 
The proof is by contradiction, i.e.\ we assume that there are two edges $uv$ and $xy$ such that $u,x \in A(Q)$ and $v \in A(P)$ and $y \in A(R)$ for some distinct free components $P,R\ne Q$. We consider two cases.

\mycase{A:} $u,x\in V(Q)$.
If $v \in A_1$ then we rotate $F(v)$ at $v$.
Similarly, if $y \in A_1$ then we rotate $F(y)$ at $y$.
Note that if both $v\in A_1$ and $y\in A_1$ then the fans $F(v)$ and $F(y)$ are distinct by Lemma~\ref{lem:usually-fans-do-not-meet}.
It follows that if both $v\in A_1$ and $y\in A_1$ then rotating $F(v)$ does not destroy $F(y)$ and we can indeed perform both rotations.
As a result, $v,y\in A_0$, which is a contradiction with Corollary~\ref{cor:C1}.

\mycase{B:} case A does not apply. 
However, since there are at least $|A(Q)|-1$ bad edges, and each vertex of $A(Q)$ is incident with at most one of them by Lemma~\ref{lem:QR-edges-matching}, we infer that at most one vertex of $V(Q)$ is not incident with a bad edge. 
Since Case A does not apply, for some free component $P\ne Q$ each bad edge incident with $V(Q)$ has the other endpoint in $A(P)$. 
If Claim 1 does not hold, there is a bad edge $uv$, $u\in A_1(Q)$ and $v\in A(R)$, for some $R\ne Q,P$. 
Then we rotate $F(u)$ at $u$ and the component $Q$ is replaced by a new component $Q'$ with $u\in V(Q')$. 
Since $|V(Q)|\ge 3$, $|V(Q)\setminus V(Q')|\le 1$ and at most one vertex of $V(Q)$ is not incident with a bad edge it means that at least one vertex of $V(Q)\cap V(Q')$ is incident with a bad edge in the new colored graph, and every such bad edge has an endpoint in $A(P)$. 
However, then we proceed as in Case A (note that by Lemma~\ref{lem:usually-fans-do-not-meet} rotating the fan $F(u)$ does not affect the fans $F(v)$ and $F(y)$). 
This finishes the proof of Claim 1.

Let $P$ be the free component from Claim 1.

\noindent {\bf Claim 2:} There is at most one bad edge incident both with $A(Q)$ and $V(P)$.

\noindent {\em Proof of Claim 2.} 
Assume there are two such edges, say $q_1p_1$ and $q_2p_2$ with $q_1,q_2\in A(Q)$. 
Since there are at least $|A(Q)|-1$ bad edges and $|V(Q)|\ge 3$, there are also at least two bad edges incident both with $V(Q)$ and with $A(P)$, say $q_3p_3$ and $q_4p_4$ with $q_3,q_4\in V(Q)$. 
By Lemma~\ref{lem:QR-edges}, $q_1,q_2\in A_1(Q)$ and $p_3,p_4\in A_1(P)$.
Then we rotate $F(q_1)$ at $q_1$ and let $Q'$ be the component that replaces $Q$. 
Note that at least one of $q_3$, $q_4$ is in $V(Q')$ and $|\f(Q')|\ge\Delta-1$ by Proposition~\ref{prop:free-colors}.
By symmetry assume $q_3\in V(Q')$.
First assume $Q$ is a tree. Then by Lemma~\ref{lem:A_1}$(i)$ and Lemma~\ref{lem:cycles} rotating $F(q_1)$ does not affect $F(q_2)$ so in particular $q_2\in A_1(Q')$. We see that $q_1p_1$, $q_2p_2$ and $q_3p_3$ are $S(Q')S(P)$-edges, $q_1,q_3\in V(Q')$ and $p_1,p_2\in V(P)$. This is a contradiction with Lemma~\ref{lem:QR-edges}.
Now assume $Q$ is not a tree. By Corollary~\ref{cor:num-uncolored} we have $|E(Q)|\le\floor{\Delta/2}$, so $Q$ is a 3-cycle.
Then $q_3,q_4\in V(Q')$. 
If after rotating $F(p_3)$ at $p_3$ the component $P'$ that replaces $P$ contains $p_1$, we do rotate $F(p_3)$ at $p_3$. As a result, we get two $V(Q)V(P')$-edges, namely $q_1p_1$ and $q_3p_3$; a contradiction with Corollary~\ref{cor:C1}.
Hence we can assume that after rotating $F(p_3)$ at $p_3$ the component that replaces $P$ does not contain $p_1$.
Hence $P$ is a tree and $F(p_3)$ is a $(v,p_1)$-fan for some $v\in V(P)$.
By Lemma~\ref{lem:A_1}$(i)$ we see that $F(p_4)$ is not a $(w,p_1)$-fan for any $w\in V(P)$.
It follows that after rotating $F(p_4)$ at $p_4$ the component $P'$ that replaces $P$ contains $p_1$. We get a contradiction with Corollary~\ref{cor:C1} as before. This finishes the proof of Claim 2.

The value of $\Delta$ is odd, so by Corollary~\ref{cor:num-uncolored} we have $|E(P)|\le (\Delta-1)/2$. 
Hence by Lemma~\ref{lem:A_1} we have $|A_1(P)|\le (\Delta-1)/2$, so by Lemma~\ref{lem:QR-edges-matching} there are at most $(\Delta-1)/2$ bad edges not incident with $V(P)$.
This, together with Claim 2 implies that the total number of bad edges is at most $(\Delta+1)/2$, which is at most $3$ when $\Delta=5$ and at most $4$ when $\Delta=7$.
This is a contradiction with our assumption that there are at least $|A(Q)|-1$ bad edges.
Indeed, if $\Delta=5$ then by Lemma~\ref{lem:A_1} we have $|A(Q)|-1=4$, and if $\Delta=7$ then by Lemma~\ref{lem:A_1} we have $|A(Q)|-1= 5$.
\end{proof}

\begin{lemma}
\label{lem:bound-5-colors-2-edges}
 Let $\Delta=5$ and let $Q$ be a 2-edge free component.
 Then, $\frac{\ch(A(Q))}{\ch(A(Q))+|E(Q)|} \ge \frac{23}{27}$. 
\end{lemma}

\begin{proof}
 We show that $\ch(Q) \ge \frac{23}2$, which implies the claim.
 
 By Corollary~\ref{cor:num-uncolored-0}, $|\f(Q)|\ge 4$ and by Lemma~\ref{lem:A_1}, $|A(Q)|= 5$.
 
 Assume $|\f(Q)|=5$.
 Then by Proposition~\ref{prop:bound-full-component} and Lemma~\ref{lem:degrees-full-component} , $\ch(A(Q)) \ge 5|A(Q)| - 1 - |E(G[A(Q)])| - |E(Q)|$.
 Since $|E(G[A(Q)])|\le {5 \choose 2}$ we get $\ch(A(Q))\ge 12$, as required.
 
 Finally assume $|\f(Q)|=4$. Let $D$ be the set of colored edges incident with $A(Q)$.
 By Lemma~\ref{lem:degrees-full-component}, all vertices of $A(Q)$ are of degree $\Delta$ in $G$.
 Hence, $|D|=5|A(Q)|-|E(G[A(Q)])|-2$ and by Lemma~\ref{lem:bound-almost-full-component} we get $\ch(A(Q))\ge \frac{9}2|A(Q)|-|E(G[A(Q)])|-1$.
 Since $|E(G[A(Q)])|\le {5 \choose 2}$ we get $\ch(A(Q))\ge \frac{23}2$, as required.
\end{proof}

\begin{corollary}
\label{cor:main-5}
Every simple graph $G$ of maximum degree $5$ has a $5$-edge-colorable subgraph with at least $\frac{23}{27}|E|$ edges.
\end{corollary}

\begin{proof}
 By Corollary~\ref{cor:num-uncolored} every free component of a partially 5-edge-colored graph which maximizes the potential $\Psi$ has at most two edges.
 Hence, by Corollary~\ref{cor:bound-1-edge} and Lemma~\ref{lem:bound-5-colors-2-edges} the claim follows.
\end{proof}

\begin{lemma}
 \label{lem:full-fan-edges-0}
 For every free component $Q$, if $|E(Q)|\ge 2$ and $Q$ is a tree then the number of edges incident with $Q$ which are full edges of some stable fan is at least $\displaystyle\sum_{\substack{\ell \in V(Q)\\ \deg_Q(\ell)=1}}|\f(\ell)|$. 
\end{lemma}

\begin{proof}
Consider an arbitrary leaf $\ell$ of $Q$ and let $x\ell$ be the edge of $Q$ incident with $\ell$.
Pick any maximal $(x,\ell)$-fan $F_{\ell}$. Since $\ell$ is a leaf $F_{\ell}$ is stable.
By Corollary~\ref{cor:num-full}, $F_{\ell}$ has at least $|\f(\ell)|$ full edges.
By Lemma~\ref{lem:cycles}, for two different leaves $\ell$ and $\ell'$ the ends of $F_{\ell}$ and $F_{\ell'}$ are disjoint (note that we can apply the lemma since $\ell$ and $\ell'$ are not the endpoints of the same edge). 
Hence the claim follows.
\end{proof}

\begin{lemma}
\label{lem:bound-incident}
 Let $Q$ be a 2-edge free component and assume that $\Delta\in\{6,7\}$.
 Then, the charge received by $A(Q)$ from the edges incident with $Q$ is at least $\frac{3}2\Delta+1-4\epsilon_\Delta$.
\end{lemma}

\begin{proof}
 Similarly as in the proof of Lemma~\ref{lem:bound-incident-1-edge}, there are exactly $\sum_{v\in Q}(\Delta-\f(v))$ colored edges incident with $Q$ and each of them sends at least $\frac{1}{2}$ to $Q$ by Lemma~\ref{lem:incident}.

For every vertex $y\in V(Q)$, for every incident uncolored edge $xy \in E(Q)$ we choose a maximal $(x,y)$-fan and it has at least $|\f(y)|$ full edges by Corollary~\ref{cor:num-full}. It follows that there are at least $\sum_{v\in V(Q)}\deg_Q(v)|\f(v)|$ full fan edges incident with $V(Q)$, and by Lemma~\ref{lem:incident}  each of them sends at least $1-\epsilon_\Delta$ to $Q$,
 
 The component $Q$ is a 2-path, say $pqr$, where $|\f(p)|,|\f(r)|\ge 1$ , and $|\f(q)|\ge 2$. 
 By Lemma~\ref{lem:full-fan-edges-0} there are at least $|\f(p)|+|\f(r)|$ edges incident with $Q$ that are full edges of stable fans, and by Lemma~\ref{lem:incident} each of them sends $1$ to $Q$.
 
The charge sent from the edges incident with $Q$ to $A(Q)$ is at least 
 \begin{eqnarray*}
 &&\tfrac{1}{2}\sum_{v\in Q}(\Delta-\f(v)) + (\tfrac{1}{2}-\epsilon_\Delta)\sum_{v\in Q}\deg_Q(v)|\f(v)| + \epsilon_\Delta(|\f(p)|+|\f(r)|) = \\
 &&\frac{\Delta|V(Q)|}2 + \tfrac{1}{2}\underbrace{\sum_{v\in Q}|\f(v)|(\deg_Q(v)-1)}_{|\f(q)|} - \epsilon_\Delta\underbrace{\sum_{v\in Q}\deg_Q(v)|\f(v)|}_{|\f(p)|+2|\f(q)|+|\f(r)|} + \epsilon_\Delta(|\f(p)|+|\f(r)|) = \\
 &&\tfrac{3}2\Delta + (\tfrac{1}{2}-2\epsilon_\Delta)|\f(q)| \ge \tfrac{3}2\Delta + 1 - 4\epsilon_\Delta.
 \end{eqnarray*}

 \end{proof}

\begin{corollary}
\label{cor:6-7-colors-2-edges}
Let $Q$ be a free component of $(G,\pi)$.
If $|E(Q)|=2$ and $\Delta\in\{6,7\}$, then
\[\frac{\ch(A(Q))}{\ch(A(Q))+|E(Q)|} \ge \begin{cases}
                                            \frac{19}{22} & \text{when $\Delta=6$,}\\ 
                                            \frac{211}{239}>\frac{22}{25} & \text{when $\Delta=7$.} 
                                         \end{cases}
\]
\end{corollary}

\begin{proof}
By Lemma~\ref{lem:A_1} we have $|A_1(Q)|= 2$.
Hence, by Lemma~\ref{lem:bound-incident} and Proposition~\ref{prop:bound-full-end} we have
$\ch(A(Q)) \ge \tfrac{3}2\Delta + 1 - 4\epsilon_\Delta + 2\cdot \half \cdot (\Delta-3)$.
Hence, $\ch(A(Q))\ge 12\frac{2}3$ if $\Delta=6$ and $\ch(A(Q))\ge 15\frac{1}{14}$ if $\Delta=7$.
Then $\ch(A(Q))/(\ch(A(Q))+|E(Q)|)$ is at least $\frac{19}{22}$ for $\Delta=6$ and $\frac{211}{239}$ for $\Delta=7$, as required.
\end{proof}

\begin{lemma}
\label{lem:6-colors-3-edges}
 Let $\Delta=6$ and let $Q$ be a 3-edge free component. 
 Assume that $G$ does not contain a set of 6 vertices $S$ such that $G[S]$ induces a clique and exactly 6 edges leave $S$.
 Then, $\frac{\ch(A(Q))}{\ch(A(Q))+|E(Q)|} \ge \frac{19}{22}$.
 \end{lemma}

\begin{proof}
By Corollary~\ref{cor:num-uncolored-0} we have $|\f(Q)|= 6$ and by Lemma~\ref{lem:A_1} we have $|A(Q)|=6$.
By Corollary~\ref{cor:bound-full-even-component} we have $\ch(A(Q)) \ge 6|A(Q)| - |E(G[A(Q)])| - 3$.
If $G[A(Q)]$ does not induce a clique then $|E(G[A(Q)])|\le {6 \choose 2} - 1$ and hence $\ch(A(Q)) \ge 19$.
Finally, if $G[A(Q)]$ induces a clique then since every vertex in $A(Q)$ is of degree 6 there are exactly 6 edges leaving $A(Q)$; a contradiction.
It follows that $\frac{\ch(A(Q))}{\ch(A(Q))+|E(Q)|}\ge\frac{19}{22}$, as required.
\end{proof}

In the following lemma by {\em extending} a partial coloring $\pi$ we mean finding a new coloring which matches $\pi$ at the edges already colored in $\pi$.

\begin{lemma}
\label{lem:6-colors-extend}
Let $G$ be a graph of maximum degree 6 that contains a subgraph $H$ isomorphic to a 6-clique.
Let $\pi$ be an arbitrary partial 6-edge-coloring of $G$ such that the edges of $H$ are uncolored.
Assume there are two vertices $v,w$ of $H$ such that $|\f(v)\cap\f(w)|\ge 5$.
Then, $\pi$ can be extended so that at most 2 edges of $H$ are left uncolored.
\end{lemma}

\begin{proof}
 Let $V(H)=\{v,w,x_1,\ldots,x_4\}$.
 Note that for every $i=1,\ldots,4$ we have $\f(x_i)\ge 5$.
 Assume w.l.o.g.\ that $\{1,\ldots,5\}\subseteq \f(v)\cap\f(w)$.
 If among $x_1,\ldots,x_4$ there are at most two vertices incident with an edge colored with a color from $\{1,\ldots,5\}$ then we just color $H$ with colors $1,\ldots,5$ using Lemma~\ref{lem:odd-clique}. As a result, at most two edges of $H$ get the same color as an incident edge so we can uncolor these two edges and get the claim.
 Otherwise, by Lemma~\ref{lem:odd-clique} and by the symmetry we can color $E(H)$ with colors $1,\ldots,5$ so that $\pi(x_1x_2) \in (\pi(x_1)\cup\pi(x_2))\setminus\{6\}$ and $\pi(x_3x_4) \in (\pi(x_3)\cup\pi(x_4))\setminus\{6\}$. Next we recolor $x_1x_2$ and $x_3x_4$ to color 6. Clearly, then at most two edges of $H$ still have the same color as an incident edge so we can uncolor these two edges and get the claim.
\end{proof}

Now we are ready to finish the proof of our bound for 6 colors. Similarly as in the case of 4 colors we need to exclude some special case when $G$ contains a dense structure, which unfortunately turns out to be quite technical this time.

\begin{lemma}
\label{lem:main-6}
Every connected simple graph $G$ of maximum degree at most $6$ has a $6$-edge-colorable subgraph with at least $\frac{19}{22}|E|$ edges, unless $G=K_7$.
\end{lemma}

\begin{proof}
 We use the induction on $|E(G)|$. 
 For the base case observe that the claim holds for the empty graph.
 Now we proceed with the induction step.
 
 First assume that $G$ does not contain a set of 6 vertices $S$ such that $G[S]$ induces a clique and exactly 6 edges leave $S$. 
 By Corollary~\ref{cor:num-uncolored} every free component of a partially 6-edge-colored graph which maximizes the potential $\Psi$ has at most three edges.
 Hence, by Corollary~\ref{cor:bound-1-edge}, Corollary~\ref{cor:6-7-colors-2-edges} and Lemma~\ref{lem:6-colors-3-edges} the claim follows.
 
 Hence in what follows we assume that there is a set $S\subset V(G)$ such that $G[S]$ induces a $K_6$ and exactly 6 edges leave $S$ (each vertex of $S$ is incident with one of them).
 
 Now assume that there are two edges leaving $S$, say $vx$ and $wy$ with $v,w\in S$, such that $x\ne y$ and $xy\not\in E(G)$.
 Then we remove $S$ from $G$ and add edge $xy$. Next we apply the induction hypothesis to the resulting graph $G'$, getting a partial coloring $\pi'$.
 We color $E(G)\cap E(G')$ according to $\pi'$, and we color $vx$ and $wy$ with $\pi'(xy)$. Next we color the remaining 4 edges leaving $S$ with free colors and we color $E(G[S])$ using Lemma~\ref{lem:6-colors-extend} so that at most two edges are left uncolored. As a result we get a partial coloring where the number of colored edges is at least $\frac{19}{22}|E(G')| + 18 = \frac{19}{22}|E(G')| + \frac{18}{20}(|E(G)|-|E(G')|) > \frac{19}{22}|E(G)|$, as required.

 Hence we can assume that $N(S)$ induces a clique. Since $G\ne K_7$, $|N(S)|>1$.
 Let $N(S)=\{v_1,\ldots,v_{|N(S)|}\}$.
 We remove the edges of $E(N[S])$ and we (partially) color the resulting graph $G'$ inductively.
 In what follows we show (for each value of $|N(S)|$ separately) that the coloring $\pi'$ of $G'$ can be extended to a coloring $\pi''$ so that (1) at most one edge of $E(G[N(S)])$ is uncolored and (2) there are two vertices $v,w\in N(S)$ such that $\f(v)\cap \f(w) \ne \emptyset$.
 Having that, we extend the coloring further. We pick an edge $vx$ for $x\in S$ and we color it with a color $a\in\f(v)\cap \f(w)$.
 Next we pick an edge $wy$ for $y\in S$ and we color it with the same color $a$. The remaining edges of $E(N(S),S)$ are colored with free colors.
 Note that $|\f(x)\cap\f(y)|=5$. Finally we partially color $G[S]$ using Lemma~\ref{lem:6-colors-extend} so that at most 2 edges remain uncolored.
 As a result we get a partial coloring of $G$ where the number of colored edges is at least $\frac{19}{22}|E(G')| + {|N(S)| \choose 2} - 1 + 6 + {6\choose 2} - 2 \ge \frac{19}{22}|E(G')| + \frac{{|N(S)| \choose 2} + 18}{{|N(S)| \choose 2} + 21}(|E(G)|-|E(G')|) \ge \frac{19}{22}|E(G)|$, as required.

 \mycase{1:} $|N(S)|=6$.
 Then $E(N(S),V\setminus N[S])=\emptyset$. We color $E(G[N(S)])$ with colors $1,\ldots,5$ according to Lemma~\ref{lem:odd-clique}. 
 Then for every $v,w\in N(S)$ we have $\f(v)\cap \f(w) =\{6\}$.

\mycase{2:} $|N(S)|=5$.
W.l.o.g.\ we can assume that for every $i=1,\ldots,4$ we have $|E(\{v_i\},S)|=1$ and $|E(\{v_5\},S)|=2$. 
Then, for every $i=1,\ldots,4$ we have $|E(\{v_i\},V\setminus N[S])|\le 1$ and $E(\{v_5\},V\setminus N[S])|=0$

First assume that there is a color, say color 1, which appears at all the four edges of $E(N(S),V\setminus N[S])$.
Then by Lemma~\ref{lem:even-clique} we can color $G[N(S)]$ with colors $2,\ldots,5$ so that only $v_1v_2$ and $v_3v_4$ are uncolored. Then we color $v_1v_2$ with 6 and in the resulting coloring $\pi$ we have $6\in\f(v_3)\cap\f(v_4)$, so $\f(v_3)\cap\f(v_4)\ne \emptyset$, as required.

Now w.l.o.g.\ we can assume that $\pi'(v_1)=\{1\}$, $\pi'(v_2)=\{2\}$ and $\pi'(v_3),\pi'(v_4)\subset\{1,\ldots,4\}$ (recall that $|\pi'(v_3)|=|\pi'(v_4)|=1$).
Then by Lemma~\ref{lem:even-clique} we can color $G[N(S)]$ with colors $1,\ldots,4$ so that exactly two edges are uncolored and they form a matching, color 1 is not used at the edges of $G[N(S)]$ incident with $v_1$ and color 2 is not used at the edges of $G[N(S)]$ incident with $v_2$. Then there are at most two edges in $G[N(S)]$ which are colored with the same color as an incident edge (each of $v_3$, $v_4$ is incident with at most one such edge). We uncolor these edges. Hence $G[N(S)]$ has at most four uncolored edges and two of them form a matching. We color these two edges with 5 and one of the remaining two (if any) with 6. Hence we get a proper partial coloring with at most one edge of $G[N(S)]$ uncolored and such that 6 is free in at least two of vertices $v_1,\ldots,v_4$, as required.

\mycase{3:} $|N(S)|=4$.
Note that for every $i=1,\ldots,4$ we have $1\le |E(\{v_i\},S)| \le 3$.
Since $|E(\{v_1,v_2,v_3,v_4\},S)|=6$ there are two subcases to consider.

\mycase{3.1:} $N(S)$ has a vertex, say $v_4$, such that $|E(\{v_4\},S)| = 3$. 
Then for every $i=1,\ldots,3$ we have $|E(\{v_i\},S)|=1$ and $|E(\{v_i\},V\setminus N[S])|\le 2$.
Moreover, $|E(\{v_4\},V\setminus N[S])|=0$.

First assume that the edges of $E(N(S),V\setminus N[S])$ use at most 5 colors (say, colors $1,\ldots,5$).
Then at least one pair of vertices from $\{v_1,v_2,v_3\}$ is incident with edges of at most 3 colors, by symmetry we can assume $v_1,v_2$ is such a pair.
We color $v_1v_3$ and $v_2v_4$ with 6.
Then $|\f(v_2v_3)|\ge 1$, $|\f(v_1v_2)|\ge 2$, $|\f(v_1v_4)|\ge 3$ and $|\f(v_3v_4)|\ge 3$, so we can color $v_2v_3$, $v_1v_2$, and $v_1v_4$ in this order, always using a free color. We see that $v_3v_4$ still has at least one free color so $\f(v_3)\cap\f(v_4)\ne\emptyset$ as required.

Now assume that the edges of $E(N(S),V\setminus N[S])$ use all 6 colors. W.l.o.g.\ $\pi(v_1)=\{1,2\}$, $\pi(v_2)=\{3,4\}$, $\pi(v_3)=\{5,6\}$.
Then we color $v_2v_3$ with 1, 
              $v_2v_4$ with 2, 
              $v_1v_3$ with 3, 
              $v_3v_4$ with 4, 
              $v_1v_2$ with 5, and we get the color 6 free at $v_1$ and $v_4$, as required.

\mycase{3.2:} $|E(\{v_1\},S)| = |E(\{v_2\},S)| = 1$ and $|E(\{v_3\},S)| = |E(\{v_4\},S)| = 2$ (if Case 3.1 does not apply all the other cases are symmetric). 
Then $|E(\{v_1\},V\setminus N[S])|\le 2$, $|E(\{v_2\},V\setminus N[S])|\le 2$, $|E(\{v_3\},V\setminus N[S])|\le 1$ and $|E(\{v_4\},V\setminus N[S])|\le 1$.

First assume that the edges of $E(N(S),V\setminus N[S])$ use at most 5 colors (say, colors $1,\ldots,5$).
We color $v_1v_2$ and $v_3v_4$ with 6. Then we are left with coloring of the 4-cycle $v_1v_4v_2v_3v_1$ and each of its edges has at least two free colors.
Since even cycles are 2-edge-choosable~\cite{erdos1979choosability}, we can color it. Finally we uncolor one edge, say $v_1v_2$ and we get $\f(v_1)\cap\f(v_2)\ne\emptyset$, as required.

Now assume that the edges of $E(N(S),V\setminus N[S])$ use all 6 colors. W.l.o.g.\ $\pi(v_1)=\{1,2\}$, $\pi(v_2)=\{3,4\}$, $\pi(v_3)=\{5\}$ and $\pi(v_4)=\{6\}$.
Then we color $v_2v_3$ with 1, 
              $v_2v_4$ with 2, 
              $v_3v_4$ with 3, 
              $v_1v_4$ with 4, 
              $v_1v_2$ with 5, and we get the color 6 free at $v_1$ and $v_3$, as required.

\mycase{4:} $|N(S)|=3$.
For every $i=1,2,3$ we have $1\le |E(\{v_i\},S)| \le 4$. 
Assume w.l.o.g.\ that $|E(\{v_1\},S)|\le|E(\{v_2\},S)|\le|E(\{v_3\},S)|$.
There are two subcases to consider.

\mycase{4.1:} $|E(\{v_1\},S)|=1$.
Then $|E(\{v_1\},V\setminus N[S])|\le 3$.
We have either $|E(\{v_2\},S)|=1$ and $|E(\{v_3\},S)|=4$ or $|E(\{v_2\},S)|=2$ and $|E(\{v_3\},S)|=3$.
In the former case we have $|E(\{v_2\},V\setminus N[S])|\le 3$ and $|E(\{v_3\},V\setminus N[S])|=0$. 
In the latter case we have $|E(\{v_2\},V\setminus N[S])|\le 2$ and $|E(\{v_3\},V\setminus N[S])|\le 1$. 
Hence in both cases $|\f(v_2v_3)|\ge 3$ and $|\f(v_1v_3)|\ge 2$.
We color $v_2v_3$ and $v_1v_3$ with free colors and we still have at least one free color at $v_2v_3$, so $\f(v_2)\cap\f(v_3)\ne\emptyset$, as required.

\mycase{4.2:} $|E(\{v_1\},S)|\ge 2$.
Then $|E(\{v_1\},S)|=|E(\{v_2\},S)|=|E(\{v_3\},S)|=2$.
Hence, for every $i=1,2,3$ we have $|E(\{v_i\},V\setminus N[S])|\le 2$.
Hence $|\f(v_1v_2)|,|\f(v_2v_3)|,|\f(v_1v_3)|\ge 2$.
If $|\f(v_1v_2)|=|\f(v_2v_3)|=|\f(v_1v_3)|=2$ then the sets $\f(v_1v_2)$, $\f(v_2v_3)$ and $\f(v_1v_3)$ are pairwise disjoint so we just color $v_1v_2$ and $v_2v_3$ with free colors and $|\f(v_1)\cap\f(v_3)|\ge 2$.
Otherwise one of these sets, say $\f(v_1v_2)$, has cardinality at least 3. Then we color $v_2v_3$ and $v_1v_3$ with free colors and $v_1v_2$ still has a free color so $\f(v_1)\cap\f(v_2)\ne\emptyset$.

\mycase{5:} $|N(S)|=2$. We just put $\pi=\pi'$. Note that $|E(N(S),V\setminus N[S])|\le 2\cdot 6 - 2 - 6 = 4$. It follows that $|\f(v_1v_2)|\ge 2$, so $\f(v_1)\cap\f(v_2)\ne\emptyset$, as required.
\end{proof}

\begin{lemma}
\label{lem:7-colors-3-edges}
 Let $\Delta=7$ and let $Q$ be a 3-edge free component. Then, $\frac{\ch(A(Q))}{\ch(A(Q))+|E(Q)|} \ge \frac{22}{25}$.
 \end{lemma}

\begin{proof}
By Corollary~\ref{cor:num-uncolored-0} we have $|\f(Q)|\ge 6$ and by Lemma~\ref{lem:A_1} we have $|A(Q)|= 6$.
Let $D$ be the set of colored edges incident with $A(Q)$.

Assume $|\f(Q)|=7$. 
By Lemma~\ref{lem:degrees-full-component} there are at least $7|A(Q)| -1 - {|A(Q)|\choose 2}$ edges incident with $A(Q)$, so $|D|\ge 7|A(Q)| -1 - {|A(Q)|\choose 2} -3= 23$. This, together with Proposition~\ref{prop:bound-full-component} gives the claim.

Finally assume $|\f(Q)|=6$. 
By Lemma~\ref{lem:degrees-full-component} there are at least $7|A(Q)| - {|A(Q)|\choose 2}$ edges incident with $A(Q)$, so $|D|\ge 7|A(Q)| - {|A(Q)|\choose 2} -3$ and by Lemma~\ref{lem:bound-almost-full-component} we have $\ch(Q)\ge \frac{13}2|A(Q)|- {|A(Q)|\choose 2} -2= 22$.
This gives the claim.
\end{proof}

\begin{corollary}
\label{cor:main-7}
Every simple graph $G$ of maximum degree $7$ has a $7$-edge-colorable subgraph with at least $\frac{22}{25}|E|$ edges.
\end{corollary}

\begin{proof}
 By Corollary~\ref{cor:num-uncolored} every free component of a partially 7-edge-colored graph which maximizes the potential $\Psi$ has at most three edges.
 Hence, by Corollary~\ref{cor:bound-1-edge}, Corollary~\ref{cor:6-7-colors-2-edges} and Lemma~\ref{lem:7-colors-3-edges} the claim follows.
\end{proof}

\section{Approximation Algorithms}
\label{sec:OTW}

In this section we describe a meta-algorithm for the maximum $k$-edge-colorable subgraph problem.
It is inspired by a method of Kosowski~\cite{K09} developed originally for $k=2$.
In the end of the section we show that the meta-algorithm yields new approximation algorithms
 for $k=3$ in the case of multigraphs and for $k=3,\ldots,7$ in the case of  simple graphs.

Throughout this section $G=(V,E)$ is the input graph from a family of graphs $\G$ (later on, we will use $\mathcal{G}$ as the family of all simple graphs or of all multigraphs).
We fix a maximum $k$-edge-colorable subgraph $\OPT$ of $G$.

As many previous algorithms, our method begins with finding a maximum $k$-matching $F$ of $G$ in polynomial time.
Clearly, $|E(\OPT)| \le |E(F)|$.
Now, if we manage to color $\rho|E(F)|$ edges of $F$, we get a $\rho$-approximation.
Unfortunately, this way we can get a low upper bound on the approximation ratio. 
Consider for instance the case of $k=3$ and $\G$ being the family of multigraphs. 
Then, if  a connected component $Q$ of $F$ is isomorphic to $G_3$, we get $\rho \le \frac{3}{4}$.
In the view of Corollary~\ref{corollary-7-9} this is very annoying, since $G_3$ is the only graph which prevents us from obtaining the $\frac{7}{9}$ ratio there.
However, we can take a closer look at the relation of $Q$ and $\OPT$.
Observe that if $\OPT$ does not leave $Q$, i.e.\ $\OPT$ contains no edge with exactly one endpoint in $Q$ then $|E(\OPT)|=|E(\OPT[V\setminus V(Q)])|+|E(\OPT[V(Q)])|$. Note also that $|E(\OPT[V(Q)])|=3$, so if we take only three of the four edges of $Q$ to our solution we do not lose anything --- locally our approximation ratio is 1. It follows that if there are many components of this kind, the approximation ratio is better than $3/4$.
What can we do if there are many components isomorphic to $G_3$ with an incident edge of $\OPT$?
The problem is that we do not know $\OPT$. However, then there are many components isomorphic to $G_3$ with an incident edge of the input graph $G$. 
The idea is to add some of these edges in order to form bigger components (possibly with maximum degree bigger than $k$) which have larger $k$-colorable subgraphs than the original components.

In the general setting, we consider a family graphs $\F \subset \G$ such that for every graph $A\in \F$,
\begin{enumerate}[(F1)]
 \item $\Delta(A)=k$ and $A$ has at most one vertex of degree smaller than $k$,
 \item for every graph $G\in \G$, for every subgraph $H$ of $G$ with $|V(A)|$ vertices, $c_k(H)\le c_k(A)$,
 \item a maximum $k$-edge colorable subgraph of $A$ (together with its $k$-edge-coloring) can be found in polynomial time; similarly, for every edge $uv\in E(A)$ a maximum $k$-ECS of $A-uv$ (together with its $k$-edge-coloring) can be found in polynomial time,
 \item for a given graph $B$ one can check whether $A$ is isomorphic to $B$ in polynomial time,
 \item $A$ is 2-edge-connected,
 \item for every edge $uv\in A$, we have $c_k(A-uv)=c_k(A)$. 
\end{enumerate}
A family that satisfies the above properties will be called a {\em $k$-normal family}.
We assume there is a number $\alpha \in (0,1]$ and a polynomial-time algorithm $\mathcal{A}$ such that for every $k$-matching $H\not\in\F$ of a graph in $\G$,  the algorithm ${\mathcal A}$ finds a $k$-edge-colorable subgraph of $H$  with at least $\alpha |E(H)|$ edges.
Intuitively, $\F$ is a family of ``bad exceptions'' meaning that for every graph $A$ in $\F$, there is $c(A) < \alpha |E(A)|$, e.g.\ in the above example of subcubic multigraphs $\F=\{G_3\}$. We note that the family $\F$ needs not to be finite, e.g.\ in the work~\cite{K09} of Kosowski $\F$ contains all odd cycles.
We also denote
\[\beta = \min_{A,B\in\F\atop\text{$A$ is not $k$-regular}}\frac{c_k(A)+c_k(B)+1}{|E(A)|+|E(B)|+1}\text{,\quad } \gamma=\min_{A\in\F}\frac{c_k(A)+1}{|E(A)|+1}.\]
As we will see, the approximation ratio of our algorithm is $\min\{\alpha,\beta,\gamma\}$.

Let $\Gamma$ be the set of all connected components of $F$ that are isomorphic to a graph in $\F$.
\begin{observation}
\label{obs-degree-k}
Without loss of generality, there is no edge $xy\in E(G)$ such that for some $Q\in\Gamma$, $x\in V(Q)$, $y\not\in V(Q)$ and $\deg(y)<k$.
\end{observation}
\begin{proof}
If such an edge exists, we replace in $F$ any edge of $Q$ incident with $x$ with the edge $xy$. The new $F$ is still a maximum $k$-matching in $G$. By (F5) the number of connected components of $F$ increases, so the procedure eventually stops with a $k$-matching having the desired property.
\end{proof}
When $H$ is a subgraph of $G$ we denote $\Gamma(H)$ as the set of components $Q$ in $\Gamma$ such that $H$ contains an edge $xy$ with $x\in V(Q)$ and $y\not\in V(Q)$. We denote $\overline{\Gamma}(H)=\Gamma\setminus\Gamma(H)$.
The following lemma, a generalization of Lemma 2.1 from~\cite{K09}, motivates the whole approach.

\begin{lemma}
 \label{lem-approx-lower-bound}
 $\displaystyle|E(\OPT)| \le |E(F)|-\sum_{Q\in\overline{\Gamma}(\OPT)}\ovc_k(Q)$.
\end{lemma}

\begin{proof}
Since for every component $Q\in\ovGamma(\OPT)$ the graph $\OPT$ has no edges with exactly one endpoint in $Q$,
\begin{equation}
\label{eq-zuzia}
|E(\OPT)| = |E(\OPT [V'])| + \sum_{Q\in\ovGamma(\OPT)}|E(\OPT[V(Q)])|,
\end{equation}
where $V'=V \setminus \bigcup_{Q\in\ovGamma(\OPT)}V(Q)$. By (F2), we
get 
\begin{equation}
\label{eq-fredzia}
|E(\OPT[V(Q)])|\le c_k(Q). 
\end{equation}
Since $\OPT$ is $k$-edge-colorable, $E(\OPT [V'])$ is a $k$-matching. Clearly $|E(\OPT [V'])| \le |E(F[V'])|$ for otherwise $F$ is not maximal.
This, together with~\eqref{eq-zuzia} and \eqref{eq-fredzia} gives the desired inequality as follows.
\begin{equation}
|E(\OPT)| \le |E(F [V'])| + \sum_{Q\in\ovGamma(\OPT)}c_k(Q)=|E(F)|-\sum_{Q\in\overline{\Gamma}(\OPT)}\ovc_k(Q). 
\end{equation}
\end{proof}

The above lemma allows us to leave up to $\sum_{Q\in\overline{\Gamma}(\OPT)}\ovc_k(Q)$ edges of components in $\Gamma$ uncolored for free, i.e.\ without obtaining approximation factor worse than $\alpha$. In what follows we ``cure'' some components in $\Gamma$ by joining them with other components by edges of $G$. We want to do it in such a way that the remaining, ``ill'', components have a partial $k$-edge-coloring with no more than $\sum_{Q\in\overline{\Gamma}(\OPT)}\ovc_k(Q)$ uncolored edges.
To this end, we find a $k$-matching $R\subseteq G$ which satisfies the following conditions:
  \begin{enumerate}[(M1)]
   \item for each edge $xy\in R$ there is a component $Q\in \Gamma$ such that $x\in V(Q)$ and $y\not\in V(Q)$,
   \item $R$ maximizes $\sum_{Q\in\Gamma(R)}\ovc_k(Q)$,
   \item $R$ is inclusion-wise minimal $k$-matching subject to (M1) and (M2).
  \end{enumerate}

\begin{lemma}
\label{lem:R-polynomial}
 $R$ can be found in polynomial time.
\end{lemma}

\begin{proof}
We use a slightly modified algorithm from the proof of Proposition 2.2 in~\cite{K09}.
We define a graph $G'=(V',E')$ as follows.
Let $V'=V\cup\{u_Q,w_Q\ :\ Q\in\Gamma\}$.
Then, for each $Q\in\Gamma$, the set $E'$ contains three types of edges:
\begin{itemize}
 \item all edges $xy\in E(G)$ such that $x\in V(Q)$ and $y\not\in V(Q)$,
 \item an edge $vu_Q$ for every vertex $v\in V(Q)$, and
 \item an edge $u_Qw_Q$.
\end{itemize}
Next we define functions $f,g:V'\rightarrow\mathbb{N}\cup\{0\}$ as follows: 
for every $v\in \bigcup_{Q\in\Gamma}V(Q)$ we set $f(v)=1$, $g(v)=k$; 
for every $v\in V\setminus\bigcup_{Q\in\Gamma}V(Q)$ we set $f(v)=0$, $g(v)=k$; 
for every $Q\in\Gamma$ we set $f(u_Q)=0$, $g(u_Q)=|V(Q)|$ and $f(w_Q)=0$, $g(w_Q)=1$.
An {\em $[f,g]$-factor} $R'$ in $G'$ is a subgraph $R'\subseteq G'$ such that for every $v\in V(R')$ there is $f(v)\le\deg_{R'}(v)\le g(v)$.
All edges $u_Qw_Q$ have weight $\ovc_k(Q)$ while all the other edges have weight 0.
Then we find a maximum weight $[f,g]$-factor $R'$ in $G'$, which can be done in polynomial time (see e.g.~\cite{schrijver}).
It is easy to see that $R=E(R')\cap E(G)$ satisfies  (M1) and (M2).
Next, as long as $R$ contains an edge $xy$ such that $R-xy$ still satisfies (M1) and (M2), we replace $R$ by $R-xy$.
 \end{proof}

Assuming that the components from $\Gamma(R)$ will be ``cured'' by joining them to other components, the following lemma shows that we do not need to care about the remaining components, i.e. the components from $\ovGamma(R)$.
Informally, the lemma says that the number of uncolored edges in such components is bounded by the the number of uncolored edges in components in $\ovGamma(\OPT)$, which will turn out to be optimal thanks to property (F2).

\begin{lemma}
\label{lem-good-matching} 
$\displaystyle \sum_{Q\in\ovGamma(R)}\ovc_k(Q)\le\sum_{Q\in\ovGamma(\OPT)}\ovc_k(Q)$.
\end{lemma}

\begin{proof}
Let $R_{\OPT}=\{xy\in E(\OPT) : \text{for some $Q\in\Gamma$, } x\in Q \text{ and } y\not\in Q\}$. Since $\OPT$ is $k$-edge-colorable, $R_{\OPT}$ is a $k$-matching. By (M2) it follows that
\begin{equation}
\label{eq-rysia}
 \sum_{Q\in\Gamma(R)}\ovc_k(Q) \ge^{\rm (M2)} \sum_{Q\in\Gamma(R_{\OPT})}\ovc_k(Q)=\sum_{Q\in\Gamma(\OPT)}\ovc_k(Q),
\end{equation}
and next
\begin{equation}
 \sum_{Q\in\ovGamma(R)}\ovc_k(Q) = \sum_{Q\in\Gamma}\ovc_k(Q) - \sum_{Q\in\Gamma(R)}\ovc_k(Q) \le^{\eqref{eq-rysia}} \sum_{Q\in\Gamma}\ovc_k(Q) - \sum_{Q\in\Gamma(\OPT)}\ovc_k(Q) = \sum_{Q\in\ovGamma(\OPT)}\ovc_k(Q).
\end{equation}
\end{proof}
  
The following observation is immediate from the minimality of $R$, i.e.\ from condition (M3).  
  
\begin{observation}
\label{obs:stars}
 Let $H_F$ be a graph with vertex set $\{Q\ :\ Q \text{ is a connected component of $F$}\}$ and the edge set $\{PQ\ :\ \text{there is an edge $xy\in R$ incident with both $P$ and $Q$}\}$. Then $H_F$ is a forest, and every connected component of $H_F$ is a star.
\end{observation}
  
In what follows, the components of $F$ corresponding to leaves in $H_F$ are called {\em leaf components}.
Now we proceed with finding a $k$-edge-colorable subgraph $S$ of $G$ together with its coloring, using the algorithm described below.
In the course of the algorithm, we maintain the following invariants:

\begin{invariant}
\label{inv-degrees}
For every $v\in V$, $\deg_R(v) \le \deg_F(v)$.
\end{invariant}

\begin{invariant}
\label{inv-no-new-gammas}
If $F$ contains a connected component $Q$ isomorphic to a graph in $\F$, then $Q\in\Gamma$, in other words a new component isomorphic to a graph in $\F$ cannot appear.
\end{invariant}

By Observation~\ref{obs:stars}, each edge of $R$ connects a vertex $x$ of a leaf component and a vertex $y$ of another component.
Hence $\deg_R(x)=1\le\deg_F(x)$. By Observation~\ref{obs-degree-k}, initially $\deg_F(y)=k$, so also $\deg_R(y)\le\deg_F(y)$.
It follows that Invariant~\ref{inv-degrees} holds at the beginning, as well as Invariant~\ref{inv-no-new-gammas}, the latter being trivial. 
Now we describe the coloring algorithm.

\begin{enumerate} [Step 1]
  \item Begin with the graph with no edges $S=(V,\emptyset)$.
  \item \label{step-jadzia}
  As long as $F$ contains a leaf component $Q\in\Gamma$ and a component $P$, such that
  \begin{itemize}
   \item there is an edge $xy\in R$ with $x\in Q$ and $y\in P$,
   \item there is an edge $yz\in E(P)$ such that no connected component of $P-yz$ is isomorphic to a graph in $\F$,
  \end{itemize}
  then we remove $xy$ from $R$ and both $Q$ and $yz$ from $F$. 
  Notice that if $z$ was incident with an edge $zw\in R$ then by Observation~\ref{obs:stars}, $w$ belongs to another leaf component $Q'$.
  Then we also remove $zw$ from $R$ and $Q'$ from $F$ (if there are many such edges $zw$ we perform this operation only for {\em one} of them). 
  It follows that Invariants~\ref{inv-degrees} and~\ref{inv-no-new-gammas} hold.
  \item \label{step-ula}
        As long as there is a leaf component $Q\in\Gamma(R)$ we do the following.
        Let $P$ be the component of $F$ such that there is an edge $xy\in R$ with $x\in Q$ and $y\in P$.
	Then, by Step~\ref{step-jadzia}, for each edge $yz\in E(P)$ in graph $P-yz$ there is a connected component isomorphic to a graph in $\F$.
	In particular, by (F1) every edge $yz\in E(P)$ is a bridge in $P$. By (F5), $P\not\in\Gamma$.
	Let $yz$ be any edge incident with $y$ in $P$, which exists by Invariant~\ref{inv-degrees}.
	Note that if $P-yz$ has a connected component $C$ isomorphic to a graph in $\F$ and containing $y$ then every edge of $C$
	incident with $y$ is a bridge in $C$; a contradiction with (F5).
	Hence $P-yz$ has exactly one connected component isomorphic to a graph in $\F$, call it $P_{yz}$, and $V(P_{yz})$ contains $z$.
	Assume $P_{yz}$ is incident with an edge of $R$, i.e. there is an edge $x'y'$ with $x'\in V(Q')$ for some leaf component $Q'\in\Gamma(R)$ and $y'\in P_{yz}$. By the same argument, $y'$ is incident with a bridge $y'z'$ in $P$ and $P-y'z'$ contains a connected component $P'_{y'z'}$ from $\F$, such that $z'\in V(P'_{y'z'})$. But since $P_{yz}$ has no bridges, $y'=z$ and $z'=y$, which implies that $P-yz$ has {\em two} connected components isomorphic to a graph in $\F$, a contradiction. Hence $P_{yz}$ is not incident with an edge of $R$.
	Then we remove $Q$, $yz$ and $P_{yz}$ from $F$ and $xy$ from $R$.
	The above discussion shows that Invariants~\ref{inv-degrees} and~\ref{inv-no-new-gammas} hold.
   \item Process each of the remaining components $Q$ of $F$, depending on its kind.
  \begin{enumerate}
   \item If $Q\in\Gamma$, it means that $Q\in\ovGamma(R)$, because otherwise there are leaf components in $\Gamma(R)$, which contradicts Step~\ref{step-ula}.
         Then we find a maximum $k$-edge-colorable subgraph $S_Q\subseteq Q$, which is possible in polynomial time by (F3), and add it to $S$ with the relevant $k$-edge-coloring.
   \item If $Q\not\in\Gamma$ we use the algorithm $\mathcal{A}$ to color at least $\alpha|E(Q)|$ edges of $Q$ and we add the colored edges to $S$.
   \item For every $Q$, $yz$ and $P_{yz}$ deleted in Step~\ref{step-ula}, we find the maximum $k$-edge-colorable subgraph $Q^*$ of $Q$ and $P^*$ of $P_{yz}$.
         Note that the coloring of $P^*$ can be extended to $P^*+yz$ since $\deg_{P^*}(z)<k$.
	 Next we add $Q^*$, $P^*$ and $yz$ to $S$ (clearly we can rename the colors of $P^*+yz$ so that we avoid conflicts with the already colored edges incident with $y$). To sum up, we added $c_k(Q)+c_k(P_{yz})+1$ edges to $S$, which is $\frac{c_k(Q)+c_k(P_{yz})+1}{|E(Q)|+|E(P_{yz})|+1}\ge\beta$ of the edges of $F$ deleted in Step~\ref{step-ula}.
   \item For every $xy$ and $Q$ deleted in Step~\ref{step-jadzia}, let $zw$ be any edge of $Q$ incident with $x$ and then we find the maximum $k$-edge-colorable subgraph $Q^*$ of $Q-zw$ using the algorithm guaranteed by (F3). Next we add $Q^*$ and $xy$ to $S$ (similarly as before, we can rename the colors of $Q^*+xy$ so that we avoid conflicts with the already colored edges incident with $y$). By (F6), $c_k(Q-zw)=c_k(Q)$.
   Recall that in Step~\ref{step-jadzia} two cases might happen: either we deleted only $Q$ and $yz$ from $F$, or we deleted $Q$, $yz$ and $Q'$.
   In the former case we add $c_k(Q)+1$ edges to $S$, which is $\frac{c_k(Q)+1}{|E(Q)|+1}\ge\gamma$ of the edges removed from $F$.
   In the latter case we add  $c_k(Q)+c_k(Q')+2$ edges to $S$, which is $\frac{c_k(Q)+c_k(Q')+2}{|E(Q)|+|E(Q')|+1}>\gamma$ of the edges removed from $F$.
  \end{enumerate}
\end{enumerate}

\begin{proposition}
 Our algorithm has approximation ratio of $\min\{\alpha,\beta,\gamma\}$.
\end{proposition}

\begin{proof}
Let $\rho = \min\{\alpha,\beta,\gamma\}$.
 \begin{eqnarray*}
  |S| & \ge & \rho (|E(F)| - \sum_{Q\in\ovGamma(R)}|E(Q)|) + \sum_{Q\in\ovGamma(R)}c_k(Q) \ge \\
      & & \rho (|E(F)| - \sum_{Q\in\ovGamma(R)}|E(Q)| + \sum_{Q\in\ovGamma(R)}c_k(Q)) = \\
      & & \rho (|E(F)| - \sum_{Q\in\ovGamma(R)}\ovc_k(Q)) \ge^{\text{(Lemma~\ref{lem-good-matching})}} \\
      & & \rho (|E(F)| - \sum_{Q\in\ovGamma(\OPT)}\ovc_k(Q)) \ge^{\text{(Lemma~\ref{lem-approx-lower-bound})}}  \rho |E(OPT)|.
 \end{eqnarray*}
\end{proof}

\begin{theorem}
\label{th:meta-algorithm}
Let $\G$ be a family of graphs and let $\F$ be a $k$-normal family of graphs.
Assume there is a polynomial-time algorithm such that for every connected $k$-matching $H \not \in {\mathcal F}$ of a graph in ${\mathcal G}$, the algorithm finds a $k$-edge-colorable subgraph of $H$ with at least $\alpha|E(H)|$ edges and a $k$-edge-coloring of it.
Moreover, let \[\beta = \min_{A,B\in\F\atop\text{$A$ is not $k$-regular}}\frac{c_k(A)+c_k(B)+1}{|E(A)|+|E(B)|+1}\text{,\quad } \gamma=\min_{A\in\F}\frac{c_k(A)+1}{|E(A)|+1}.\]
Then, there is an approximation algorithm for the maximum $k$-ECS problem for graphs in $\G$ with approximation ratio $\min\{\alpha,\beta,\gamma\}$.
\end{theorem}

The above theorem summarizes our discussion in this section.
Now we apply it to particular cases.

\begin{theorem}
The maximum $3$-ECS problem has a $\frac{7}{9}$-approximation algorithm for multigraphs.
\end{theorem}

\begin{proof}
Let $\F=\{G_3\}$. It is easy to check that $\F$ is $3$-normal.
Now we give the values of parameters $\alpha, \beta$ and  $\gamma$ from  Theorem~\ref{th:meta-algorithm}.
By Corollary~\ref{corollary-7-9}, $\alpha=\frac{7}{9}$.
Notice that $c_3(G_3)=3$ and $|E(G_3)|=4$.
Hence, $\beta=\frac{7}{9}$ and $\gamma=\frac{4}{5}$.
By Theorem~\ref{th:meta-algorithm} the claim follows.
\end{proof}

\begin{theorem}
The maximum $3$-ECS problem has a $\frac{13}{15}$-approximation algorithm for simple graphs.
\end{theorem}

\begin{proof}
Let $\F=\{B_3\}$. It is easy to check that $\F$ is $3$-normal.
Now we give the values of parameters $\alpha, \beta$ and  $\gamma$ from  Theorem~\ref{th:meta-algorithm}.
By Corollary~\ref{corollary-13-15}, $\alpha=\frac{13}{15}$.
Notice that $c_3(B_3)=6$ and $|E(B_3)|=7$.
Hence, $\beta=\frac{13}{15}$ and $\gamma=\frac{7}{8}$.
By Theorem~\ref{th:meta-algorithm} the claim follows.
\end{proof}

\begin{theorem}
The maximum $4$-ECS problem has a $\frac{9}{11}$-approximation algorithm for simple graphs.
\end{theorem}

\begin{proof}
Let $\F=\{K_5\}$. It is easy to check that $\F$ is $4$-normal.
Now we give the values of parameters $\alpha, \beta$ and  $\gamma$ from  Theorem~\ref{th:meta-algorithm}.
By Theorem~\ref{thm:main}, $\alpha=\frac{5}{6}$.
Observe that $\beta=\infty$, since $\F$ contains only $K_5$ which is $4$-regular.
Notice that $c_4(K_5)=8$ and $|E(K_5)|=10$.
Hence, $\gamma=\frac{9}{11}$.
By Theorem~\ref{th:meta-algorithm} the claim follows.
\end{proof}

\begin{theorem}
The maximum $6$-ECS problem has a $\frac{19}{22}$-approximation algorithm for simple graphs.
\end{theorem}

\begin{proof}
Let $\F=\{K_7\}$. It is easy to check that $\F$ is $6$-normal.
Now we give the values of parameters $\alpha, \beta$ and  $\gamma$ from  Theorem~\ref{th:meta-algorithm}.
By Theorem~\ref{thm:main}, $\alpha=\frac{19}{22}$.
Observe that $\beta=\infty$, since $\F$ contains only $K_7$ which is $6$-regular.
Notice that by Lemma~\ref{lem:even-clique}, $c_6(K_7)=18$ and $|E(K_7)|=21$.
Hence, $\gamma=\frac{19}{22}$.
By Theorem~\ref{th:meta-algorithm} the claim follows.
\end{proof}

Directly from Proposition~\ref{prop:approx} and from Theorem~\ref{thm:main} we get the following corollaries.

\begin{corollary}
The maximum $5$-ECS problem has a $\frac{23}{27}$-approximation algorithm for simple graphs.
 \end{corollary}

\begin{corollary}
The maximum $7$-ECS problem has a $\frac{22}{25}$-approximation algorithm for simple graphs.
 \end{corollary}

\section{Further Work}
The most important open problem seems to be to provide answers to Questions~\ref{q-odd} and~\ref{q-even} from Section~\ref{intro:combin} for all $\Delta\ge 8$. 
We think that although our techniques (with some hard work) might be sufficient to improve the Vizing bound when $\Delta=9$ or $\Delta=10$, for large values of $\Delta$ some new ideas are needed.

It would be also interesting to improve our bounds for $\Delta\le 7$. In particular the best upper bound for even $\Delta$, and for $G\ne K_{\Delta+1}$ we are aware of is $\frac{\Delta}{\Delta+1-2/\Delta}$, attained by $K_{\Delta+1}-e$, i.e. the $K_{\Delta+1}$ with one edge removed. The lemma below provides an upper bound for odd values of $\Delta$.

\begin{lemma}
 \label{lem:upper-odd}
 For every odd value of $\Delta$ there is a graph of maximum degree $\Delta$ such that \[\gamma_\Delta(G)=\frac{\Delta+1}{\Delta+2-\tfrac{1}\Delta}.\]
\end{lemma}

\begin{proof}
Let $\Delta = 2\ell + 1$.
Begin with $K_{\Delta+1}$. Remove a matching $M$ of size $\ell$. Add a new vertex $v$ and add edges between $v$ and $V(M)$.
Denote the resulting graph by $B_\Delta$. Observe that the maximum degree of $B_\Delta$ is $\Delta$. 
We see that $|E(B_\Delta)|={\Delta+1\choose 2} + \ell$.
Consider a maximum $\Delta$-edge-colorable subgraph $H$ of $B_\Delta$.
Since each of the $\Delta$ color classes has at most $(\Delta+1)/2$ edges, $|E(H)| \le \Delta\cdot (\Delta+1)/2 = {\Delta\choose 2}$.
It is easy to see that actually $|E(H)|={\Delta\choose 2}$; just consider the coloring of $K_\Delta$ from Lemma~\ref{lem:odd-clique}, and for each of the removed edges, say $xy$, copy its color to one of the new edges incident with $xy$, say $vx$.
It follows that $\gamma_\Delta(B_\Delta)={\Delta+1\choose 2}/({\Delta+1\choose 2}+\ell)=\frac{\Delta+1}{\Delta+2-1/\Delta}$.
\end{proof}

Another interesting question is the following.

\begin{question}
\label{q:13-15}
Does $\gamma_3(G)\ge\frac{13}{15}+\varepsilon$ for some $\varepsilon>0$ when $G$ is a simple graph isomorphic neither to $B_3$ nor to the Petersen graph?
\end{question}

\section*{Acknowledgments}

We are very grateful to Adrian Kosowski for helpful remarks regarding the state-of-art of the $k$-ECS problem in multigraphs.
We also thank anonymous reviewers for careful reading and many helpful remarks.

\bibliographystyle{abbrv}

\end{document}